\documentclass[acmsmall,screen]{acmart}

\AtBeginDocument{%
  }

\setcopyright{none}
\copyrightyear{2026}
\acmYear{2026}

\acmJournal{TCPS}
\acmMonth{11}

\usepackage{amsfonts}
\usepackage{amsmath}
\usepackage{amstext}
\usepackage{amsthm}
\usepackage{listings}
\usepackage{hyperref}
\usepackage{makecell}
\usepackage{xspace}
\usepackage{tikz}
\usepackage{wrapfig}
\usepackage{booktabs}
\usepackage{multirow}

\usepackage{amssymb}

\usepackage{tikz}
\usepackage{graphicx}
\usetikzlibrary{arrows.meta}
\definecolor{roadgray}{HTML}{8E8E8E}
\definecolor{lanegray}{HTML}{A6A6A6}
\definecolor{laneyellow}{HTML}{F2B01E}
\definecolor{laneaccent}{HTML}{FF5B5B}
\definecolor{trackcyan}{HTML}{7CD2CE}
\definecolor{carblue}{HTML}{1E4BFF}
\definecolor{carred}{HTML}{D93333}

\begin{document}

\newcommand{\sayan}[1]{\textcolor{black}{#1}}
\title{Verse: Semantics and Verification of Hybrid Multi-Agent Scenarios}

\author{Yangge Li}
\email{li213@illinois.edu}
\affiliation{%
  \institution{University of Illinois Urbana Champaign}
  \city{Urbana}
  \state{Illinois}
  \country{USA}
}
\author{Katherine Braught}
\email{braught2@illinois.edu}
\affiliation{%
  \institution{University of Illinois Urbana Champaign}
  \city{Urbana}
  \state{Illinois}
  \country{USA}
}
\author{Alex Yuan}
\email{ayuan20@illinois.edu}
\affiliation{%
  \institution{University of Illinois Urbana Champaign}
  \city{Urbana}
  \state{Illinois}
  \country{USA}
}
\author{Sayan Mitra}
\email{mitras@illinois.edu}
\affiliation{%
  \institution{University of Illinois Urbana Champaign}
  \city{Urbana}
  \state{Illinois}
  \country{USA}
}








\renewcommand{\shortauthors}{Li et al.}

\newcommand{\katherine}[1]{\textcolor{black}{#1}}
\newcommand{\change}[1]{\textcolor{black}{#1}}
\newcommand{\yangge}[1]{\textcolor{black}{#1}}

\newif\ifrevisionmarkup
\revisionmarkupfalse
\newcommand{\rev}[1]{\ifrevisionmarkup\textcolor{blue}{#1}\else#1\fi}
\definecolor{revtwo}{RGB}{160,0,160}
\newcommand{\revb}[1]{\ifrevisionmarkup\textcolor{revtwo}{#1}\else#1\fi}

\newif\ifflagchecks
\flagchecksfalse
\newcommand{\tocheck}[1]{\ifflagchecks\textcolor{red}{#1}\else#1\fi}

\newcommand{\type}{{\text{type}}}
\newcommand{\val}{{\text{val}}}

\newcommand{\scenario}{{SC}}
\newcommand{\agent}{\mathcal{A}}
\newcommand{\I}{\mathcal{I}}    
\newcommand{\map}{\mathcal{M}}
\newcommand{\sensor}{\mathcal{S}}
\newcommand{\mapvert}{{V}}
\newcommand{\maplabel}{{L}\xspace}
\newcommand{\maplabcurve}{\mathit{Tk}\xspace}
\newcommand{\maplablab}{\mathit{Nx}\xspace}
\newcommand{\mapcurve}{{\omega}\xspace}
\newcommand{\mapsetcurve}{{\Omega}\xspace}

\newcommand{\strategy}{{P}\xspace}
\newcommand{\strategyvar}{{p}\xspace}

\newcommand{\agentid}{{I}}
\newcommand{\agenttl}{{TL}\xspace}
\newcommand{\Zagent}{{Z}}
\newcommand{\Zagentvar}{{z}}

\newcommand{\agentdiscstate}{{P}}
\newcommand{\Xagent}{{X}}
\newcommand{\xagent}{{x}}
\newcommand{\Dagent}{{D}}
\newcommand{\dagent}{{d}}
\newcommand{\agentcontinit}{{X_0}}
\newcommand{\agentdiscinit}{{d_0}}
\newcommand{\agentcontroller}{{G}}
\newcommand{\agentflow}{{F}}
\newcommand{\agentunsafe}{{U}}
\newcommand{\agentfullstate}{{Y}}
\newcommand{\agentobs}{O}
\newcommand{\agentobsvar}{O}

\newcommand{\agentobsval}{\textbf{o}}

\newcommand{\agentobsadd}{J}
\newcommand{\agentobsaddvar}{j}
\newcommand{\agentinvariables}{I}
\newcommand{\agentinval}{\textbf{i}}
\newcommand{\agentinvar}{i}

\newcommand{\workspace}{{W}}
\newcommand{\workspacevar}{{w}}
\newcommand{\labeltrajectory}{\textbf{TL}}

\newcommand{\ourtool}{{Verse}\xspace}

\newcommand{\stateset}{X^h}
\newcommand{\fparset}{M^h}
\newcommand{\initset}{\Theta}
\newcommand{\finitpar}{m_{\mathit{init}}^h}
\newcommand{\edgeset}{E^h}
\newcommand{\haguard}{\textbf{G}}
\newcommand{\hareset}{\textbf{R}}
\newcommand{\ha}{{H}}
\newcommand{\reals}{{\mathbb{R}}} 
\newcommand{\nnreals}{{\mathbb{R}^{\geq 0}}}
\newcommand{\parset}{S}
\newcommand{\initmode}{s_{\mathit{init}}}

\newcommand{\Dha}{\textbf{D}\xspace}
\newcommand{\dha}{\textbf{d}\xspace}
\newcommand{\Xha}{\textbf{X}}
\newcommand{\Xhaset}{\textbf{S}}
\newcommand{\xha}{\textbf{x}}
\newcommand{\Tha}{\xi}
\newcommand{\transvert}{\textbf{V}}
\newcommand{\transedge}{\textbf{E}}
\newcommand{\transelab}{{elab}}
\newcommand{\tl}{\textbf{TL}}
\newcommand{\simu}{{sim}}
\newcommand{\transitiongraph}{\textbf{G}}
\newcommand{\exec}{{\alpha}}
\newcommand{\reach}[1]{{Reach_{\ha_{#1}}}\xspace}
\newcommand{\unsafe}{\textbf{U}}

\newcommand{\agentguard}{{G}\xspace}
\newcommand{\agentreset}{{R}\xspace}

\newcommand{\post}{{post}\xspace}
\newcommand{\postcont}{\texttt{postCont}\xspace}
\newcommand{\postdisc}{\texttt{postDisc}\xspace}
\newcommand{\node}{{N}\xspace}

\newcommand{\exectree}{{Tree}\xspace}
\newcommand{\treevert}{{V}\xspace}
\newcommand{\treeedge}{{E}\xspace}

\newcommand{\condval}{{CondVal}\xspace}
\newcommand{\condvals}{{CondVals}\xspace}
\newcommand{\cache}{{C}\xspace}
\newcommand{\guardcache}{{C_g}\xspace}
\newcommand{\flowcache}{{C_f}\xspace}
\newcommand{\unknown}{{\texttt{unknown}}\xspace}
\newcommand{\sat}{{\texttt{sat}}\xspace}
\newcommand{\unsat}{{\texttt{unsat}}\xspace}

\newcommand{\none}{\texttt{none}\xspace}
\newcommand{\treedepth}{{m}\xspace}
\newcommand{\algveri}{{\texttt{verify}}\xspace}
\newcommand{\alginc}{{\texttt{verifyInc}}\xspace}
\newcommand{\algfcache}{{\texttt{flowCache}}\xspace}

\newcommand{\sd}[1]{{\langle \Xhaset_{#1},\dha_{#1}\rangle}\xspace}
\newcommand{\sdp}[1]{{\langle \Xhaset_{#1}',\dha_{#1}'\rangle}\xspace}
\newcommand{\sdpp}[1]{{\langle \Xhaset_{#1}'',\dha_{#1}''\rangle}\xspace}

\newcommand{\addsensor}{\texttt{add\_sensor}\xspace}
\newcommand{\verify}{\texttt{verify}\xspace}
\newcommand{\simulate}{\texttt{simulate}\xspace}
\newcommand{\assert}{\texttt{assert}\xspace}

\newcommand{\CA}{{Q-CA}\xspace}
\newcommand{\NPC}{{Q-NPC}\xspace}

\definecolor{bg}{rgb}{0.95,0.95,0.95}

\definecolor{pygkeyword}{HTML}{008000}
\definecolor{pygstring}{HTML}{BA2121}
\definecolor{pygcomment}{HTML}{3D7B7B}
\definecolor{pygbuiltin}{HTML}{008000}
\lstdefinestyle{versepython}{
  language=Python,
  basicstyle=\ttfamily\linespread{0.859}\selectfont,
  keywordstyle=\color{pygkeyword}\bfseries,
  keywordstyle=[2]\color{pygbuiltin},
  stringstyle=\color{pygstring},
  commentstyle=\color{pygcomment}\itshape,
  morekeywords=[2]{True,False,None,self,float,int,str,bool,any,all,len,range,print},
  numbers=left,
  numberstyle=\rmfamily\tiny,
  numbersep=12pt,
  breaklines=true,
  breakatwhitespace=false,
  postbreak=,
  columns=fullflexible,
  keepspaces=true,
  showstringspaces=false,
  upquote=true,
  aboveskip=0pt,
  belowskip=0pt,
  frame=none,
}
\lstset{style=versepython}

\begin{abstract}
Verifying cyber-physical systems with many interacting agents remains difficult because interactions are mediated by sensors, individual decisions are constrained by environment (map) structure, and the resulting mode space scales combinatorially. We present \ourtool, a framework for writing agents, maps, and sensors in a familiar programming language syntax, with a formal semantics for the resulting scenarios. A compiler translates each scenario to a hybrid automaton whose executions correspond exactly to the runs of the scenario. Reachability analysis on this automaton builds an execution tree from two post operators, one for continuous flow over a time step and one for a single agent's mode change, and checks safety assertions at every vertex. The tree expands only the modes that are reached, without enumerating the global mode space. Experiments highlight expressive modeling across cars and drones, 2D/3D maps (including OpenDRIVE imports), noisy sensors, and multiple simultaneous safety assertions. \ourtool verifies a scenario with 7 agents and 37 modes in 90 seconds on a standard laptop; it is competitive with CORA and JuliaReach on ARCH benchmarks, and on  multi-agent scenarios is an order of magnitude faster than JuliaReach at the cost of looser reachable sets.
\end{abstract}



\received{20 February 2026}
\received[revised]{18 November 2026}

\maketitle


\section{Introduction}
\label{sec:intro}

Hybrid automaton formalisms~\cite{henz:thhybat,mit:hioa} provide a foundational framework for modeling and verifying robotic~\cite{egerstedt2002hybrid}, autonomous~\cite{mitra_2021}, and cyber-physical systems~\cite{kumar2012hybrid}.
\yangge{Significant technical advances have been made in formal verification of hybrid systems over the past three decades, with tools such as Hylaa, CORA, and DryVR demonstrating scalability and applicability across linear and nonlinear benchmarks~\cite{10.1145/3302504.3311792,Althoff2015a,DryVRFanQM017}.}
Despite these successes, hybrid system verification has seen limited adoption in mainstream engineering workflows. A primary reason is a persistent usability barrier: existing tools require users to manually translate system designs into specialized formal modeling languages, often with a steep learning curve. This modeling effort is typically decoupled from system development, making it unclear whether the cost of formalization is justified by the resulting verification benefits.

Experience from large-scale software verification suggests that adoption improves dramatically when verification is integrated directly at the code level, avoiding translation to isolated modeling formalisms~\cite{10.1007/978-3-031-13185-1_1,45424,Reid2020TowardsMF}. While software verifiers such as CBMC and Dafny operate directly on implementation-level artifacts, hybrid verification tools typically require users to hand-construct explicit hybrid automata, including state machines, guards, and reset maps. This process becomes especially burdensome for multi-agent systems, where the global state space grows combinatorially, and where reuse across environments is limited.

We believe that \revb{three difficulties hold back the} verification of realistic multi-agent CPS scenarios: (1) expressing individual decisions independently of map geometry while still enforcing map-specific constraints, (2) formalizing sensor-mediated interaction without omniscient state access, and (3) maintaining verification scalability as agents and modes grow.

To address these challenges, we introduced Verse~\cite{10.1007/978-3-031-37706-8_18}, a Python-based framework for modeling and verifying hybrid multi-agent scenarios. In Verse, a system is defined as a {\em scenario\/} composed of {\em agents\/} embedded in a {\em map\/} and interacting via {\em sensors\/}. Each component is implemented as a Python module and interpreted by the Verse frontend to construct a formal hybrid automaton. The same agent code can be simulated, analyzed, and when appropriate executed on hardware, enabling a unified development and verification workflow.

Verse meets the three difficulties as follows. An agent's decision logic chooses among a few internal modes, such as cruising or changing lane, and a partial next-mode function supplied by the map decides whether the chosen change is admissible from the agent's current track. The same logic therefore runs on every map, and an agent can switch only to a lane that exists at its position. Agents never read one another's state. Each reads the observations returned by a sensor, a set-valued function of the other agents' observable variables, and an I/O compatibility criterion checks that the logic reads only variables some sensor supplies. Under bounded noise a sensor returns every observation the noise could produce, and a mode change is enabled when some observation in that set satisfies its guard. Verification builds an execution tree that expands only the mode tuples actually reached, so the global mode space is never enumerated.

This article provides the first rigorous formal semantics for the Verse modeling framework. Building on our earlier tool paper~\cite{10.1007/978-3-031-37706-8_18}, we formally define the mathematical structure of Verse scenarios, including precise interfaces for agents, maps, and sensors, and the conditions under which their composition is well-posed. Establishing these semantics is nontrivial because Verse intentionally decouples high-level agent decision logic from the low-level geometric and physical details of the environment, while still allowing agents to be instantiated across heterogeneous maps and scenarios. Our contributions are:
\begin{itemize}
    \item A formal semantics for multi-agent scenarios, with maps, sensors, and compatibility criteria.
    \item A translation from scenarios to hybrid automata, with a proof that the runs of a scenario and the executions of its automaton correspond exactly (Theorem~\ref{thm:correspondence}), and a condition under which the order in which simultaneously enabled agents move cannot matter (Proposition~\ref{prop:interleaving}).
    \item An evaluation of expressive modeling (cars and drones, 2D/3D maps, noisy sensors, three reachability engines), a comparison with CORA and JuliaReach on ARCH benchmarks, and a comparison with JuliaReach on two scalable multi-agent scenarios that measures what the architecture gains in speed and loses in precision.
\end{itemize}

\section{Preliminaries}
\label{sec:prelims}

Before detailing the \ourtool formalism, we establish the mathematical notation for variables, states, and logical quantifiers used throughout this paper.

Variables are the fundamental building blocks used to define the state of a system or its components. A typed variable $v$ is associated with a type, $\type(v)$, which defines the set of all values $v$ can take. We denote the set of continuous variables (where $\type(v) \subseteq \mathbb{R}$) as $V_{cont}$ and discrete variables (where $\type(v)$ is a finite set) as $V_{disc}$. For a set of variables $V$, a \emph{valuation} is a mapping $\textbf{v}$ that assigns to each variable $v \in V$ a value in its type. We denote a single valuation using a boldface letter $\textbf{v}$. The set of all possible valuations for $V$ is denoted by $\val(V)$. Thus, if a system is defined by variables $V$, its state space is $\val(V)$.

It is often necessary to extract the local state of a specific component from a global system state. Given a set of variables $V$ and a subset $V' \subseteq V$, the \emph{restriction} (or projection) of a valuation $\textbf{v} \in \val(V)$ to $V'$ is denoted by $\textbf{v} \lceil V'$. This restricted valuation is an element of $\val(V')$ such that $(\textbf{v} \lceil V')(v) = \textbf{v}(v)$ for all $v \in V'$.

Finally, the decision logic in \ourtool scenarios relies on checking Boolean conditions over these states. To model logic that reasons over collections of components (e.g., checking if "any" agent is close), we use first-order logic quantifiers over finite index sets. Let $\mathcal{I}$ be a finite set of indices (representing, for example, the collection of agents). We denote the universal quantifier as $\forall i \in \mathcal{I}$, which holds if the associated condition is true for all $i$ (corresponding to the Python \texttt{all} function). We denote the existential quantifier as $\exists i \in \mathcal{I}$, which holds if the condition is true for at least one $i$ (corresponding to the Python \texttt{any} function). Guards of this form are common in applications, for example, the benchmarks of Section~\ref{sec:exp:multiagent} both use one.
\section{Multi-Agent Scenarios in \ourtool}
\label{sec:scenarios}

\sayan{A central contribution of this work is a new formalism for modeling scenarios with multiple hybrid agents that co-evolve in a shared map and interact through sensor-based observations. In this section, we present the underlying mathematical concepts and describe the formalism as realized in Verse. Section~\ref{sec:scenariotoHA} then demonstrates how such scenarios can be represented as hybrid automata.}

\begin{figure}
\centering
\input{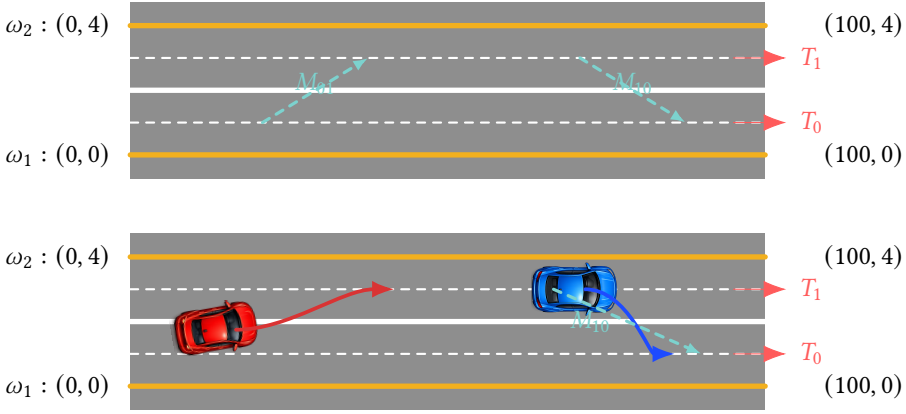}
    \caption{The two-lane highway running example. {\em Top}: The map consists of two parallel lanes defined by paths from (0,0) to (100,0) and (0,4) to (100,4). The diagram illustrates the four primary track modes: lane-following for $\omega_1$ ($\mathsf{T_0}$) and $\omega_2$ ($\mathsf{T_1}$), and the transitional maneuvers for changing from $\omega_1$ to $\omega_2$ ($\mathsf{M_{01}}$) and from $\omega_2$ to $\omega_1$ ($\mathsf{M_{10}}$).
    {\em Bottom}: The agents positions do not necessarily follow these reference tracks exactly. E.g., the blue car executes a lane change but does not precisely adhere to the linear $\mathsf{M_{10}}$ reference path.
    }
    \label{fig:example_map}
\end{figure}

Informally, a {\em scenario\/} defines how a collection of agents move within a shared  map and interact via  sensors. The map describes the constraint in the  workspace in the form of reference paths called {\em tracks\/}. Each  agent is a state machine with both discrete and continuous state variables. Our formalism models, agent-to-agent influence only at discrete transitions via sensor observations, even though the continuous dynamics proceed in parallel.
The  {\em observations} are determined by the scenario's sensor function, which might specify, for example, the relative distance to nearby agents or to static elements in the environment.

\subsection{Maps, Tracks and Agent's Internal Track Modes}
\label{sec:maps}
The {\em workspace} defines the physical environment in which the agents operate. Formally, a set of continuous variables, denoted by $\workspace$, represent position of a point in this environment, and the workspace is defined as $\val(\workspace)$.

For the Highway example  (see Figure~\ref{fig:example_map}),
the agent's \rev{workspace variables} $\workspace = \{x, y\}$ represent a vehicle's 2D position, and the  workspace $\val(\workspace)$ is therefore isomorphic to $\reals^2$.

A map defines  constraints on how an agent may move in the workspace. First of all,  a map provides a static collection of reference paths, called tracks, that the agents {\em may} follow. The tracks that an agent {\em does} follow at run time can of course also depend on its {\em internal discrete state or mode} such as ``continue straight`` or ``switch``'' to its continuous state or position in the workspace. Secondly, a map also governs the agent's internal decision-making by defining the specific preconditions, such as its current position or assigned track mode, that enable transitions between its internal modes.

\paragraph{Tracks}
A {\em track} is a continuous  function $\mapcurve: [0,1]\rightarrow \val(\workspace)$
that can  be used by an agent as a reference path.
A map $\map$ defines a set of allowed tracks $\mapsetcurve_{\map}$.
In the Highway example, there are two lanes of straight reference paths---one from position $(0, 0)$ to $(100, 0)$ and the second from $(0, 4)$ to $(100, 4)$. Formally, these two  are defined by tracks:  
\begin{equation}
\mapcurve_1(s) = \langle 100s, 0\rangle, \quad \mapcurve_2(s) = \langle 100s, 4\rangle , s\in[0,1]
\end{equation}
In addition to these straight line tracks, the Highway map $\map$ has a family of lane change tracks that go between the two lanes. One such track, shown in Figure~\ref{fig:example_map}, starts at $\langle 10, 0\rangle$ and is the linear path $\mapcurve(s) = \langle 10 + 4s, 4s \rangle$, $s \in [0,1]$, which ends at $\langle 14, 4\rangle$ on the second lane.
As we shall see, the agents, constrained by their own initial conditions and dynamics, may not follow the tracks exactly. 

As the highway example illustrates, a map may contain infinitely many tracks or reference paths. To keep agent behaviors abstract and portable across different maps, Verse groups tracks into a finite set of {\em track modes\/} $\maplabel_{\map}$. An agent thus chooses among abstract categories of motion, such as ``straight'' or ``change lane,'' rather than selecting a specific geometric path.
Each element in this set is a unique identifier (a name) for a family of related tracks. 
The set of track modes for the Highway scenario are: 
\[ \maplabel_{\map} = \{\mathsf{T_0}, \mathsf{T_1}, \mathsf{M_{01}}, \mathsf{M_{10}}\} \]
These modes represent distinct geometric categories of tracks available in the map. The modes $\mathsf{T_0}$ and $\mathsf{T_1}$ correspond to the sets of straight-line tracks, and  $\mathsf{M_{01}}$ and $\mathsf{M_{10}}$, correspond to the families of lane switching  paths.

\paragraph{Map}
Having introduced tracks and track modes, we can now define the maps precisely. 
A map $\map$ \rev{is} a $4$-tuple  $\langle \maplabel_\map, \maplabcurve, \strategy_\map, \maplablab \rangle$, where
(1)  $\maplabel_{\map}$ is the finite set of track modes, 
(2) $\maplabcurve: \val(\workspace)\times\maplabel_\map\rightarrow \mapsetcurve_\map$ outputs a concrete track, given a position and a track mode, 
(3) $\strategy_\map$ is the set of permissible internal agent modes allowed in $\map$, and 
(4)  
$\maplablab:\maplabel_\map \times \strategy_\map \times \strategy_\map \rightharpoonup \maplabel_\map$ is a {\em partial\/} function that outputs the   \sayan{ next track mode $\maplablab(\ell,p,p')$ from the current track mode $l$ and internal agent mode $p$ given that the agent has chosen to go to the next internal model $p'$. An example usage of $\maplablab$ is shown in Fig.~\ref{fig:maplablab}}. We write $\langle \ell, p, p'\rangle \in \mathrm{dom}(\maplablab)$ when the map permits the change, and the transition is forbidden otherwise; this is what prevents an agent from switching to a lane that does not exist at its current position.

\begin{figure}
    \centering
        \resizebox{0.8\linewidth}{!}{\input{figures/Tikzfigures/tikz_maplablab.tex}}
    \caption{The figure illustrating the role of $\maplablab$ function. In this example, the agent initially following track $\mathsf{T_0}$ with internal mode $\strategy=\texttt{Cruise}$. When the agent decides to change its internal mode to $\strategy'=\texttt{SwitchLeft}$, this intent, along with the current track mode and agent's internal mode, is fed into the $\maplablab$ function. The $\maplablab$ function then processes this input according to the map's rules. As a result, the agent is assigned to the transitional track mode $\mathsf{M_{01}}$.}
    \label{fig:maplablab}
\end{figure}

We now define the remaining components of the map tuple for our highway scenario:
\begin{itemize}
    \item  $\maplabcurve$ is defined as follows:
    \begin{equation} \maplabcurve(w=\langle w_1, w_2\rangle, l)=
        \begin{cases}
            \langle 100s, 0\rangle , \text{if~} l=T_0\\
            \langle 100s, 4\rangle , \text{if~} l=T_1\\ 
            \langle 4s+w_1-w_2, 4s\rangle ,\text{if~} l=M_{01}\\
            \langle 4s+w_1+w_2-4, 4-4s\rangle ,\text{if~} l=M_{10}
        \end{cases}
    \end{equation}
    where $s \in [0,1]$ is the \yangge{path parameter}, $w \in \val(\workspace)$ is the agent's position, and $l \in \maplabel_\map$ is the track mode.
    The nature of this mapping depends on the track mode. For modes $T_0$ and $T_1$, the function maps to fixed tracks that define the absolute geometry of Lane 1 and Lane 2, respectively. The selection of these tracks is therefore independent of the agent's position $w$.
    In contrast, the mapping for modes $M_{01}$ and $M_{10}$ is dependent on the agent's position $w$. This dependency ensures that the selected lane change reference path is always correctly positioned to intersect the agent's location, providing a valid connection between the two lanes.

    To make this concrete, suppose the agent is at $w = \langle 10, 0\rangle$ in the first lane and switches to $\mathsf{M_{01}}$. Then $\maplabcurve(w, \mathsf{M_{01}})(s) = \langle 10 + 4s, 4s\rangle$, so the reference path starts at $\langle 10, 0\rangle$ at $s=0$, passes through $\langle 12, 2\rangle$ at $s = 1/2$, and ends at $\langle 14, 4\rangle$ at $s=1$, on the second lane. Suppose instead the agent is at $w = \langle 10, 4\rangle$ in the second lane and switches to $\mathsf{M_{10}}$. Then $\maplabcurve(w, \mathsf{M_{10}})(s) = \langle 10 + 4s, 4-4s\rangle$, which starts at $\langle 10, 4\rangle$, passes through $\langle 12, 2\rangle$, and ends at $\langle 14, 0\rangle$ on the first lane. In both cases the path is traversed in the direction of increasing $x$, and $s = (4-w_2)/4$ recovers the agent's own position, which is what makes the reference path usable from wherever the agent happens to be when it commits to the maneuver.
    \item The set of permitted {\em internal modes} for this highway map, $\strategy_\map$, is defined as:
    \[ \strategy_\map = \{\texttt{Cruise}, \texttt{SwitchLeft}, \texttt{SwitchRight}\} \]
    These modes represent the agent's high-level intentions: to continue in the current lane (\texttt{Cruise}), to switch from Lane 1 to Lane 2 (\texttt{SwitchLeft}), or to switch from Lane 2 to Lane 1 (\texttt{SwitchRight}).
    \item For this example, $\maplablab$ is defined as:
    \begin{equation} \maplablab(l, p_0, p_1)=
    \begin{cases}
        \mathsf{T}_0 ,\text{if~} l\in\{\mathsf{T}_0, \mathsf{M}_{10}\}, p_1=\texttt{Cruise}\\ 
        \mathsf{T}_1 ,\text{if~} l\in\{\mathsf{T}_1, \mathsf{M}_{01}\}, p_1=\texttt{Cruise}\\ 
        \mathsf{M}_{01} ,\text{if~} l=\mathsf{T}_0, p_0=\texttt{Cruise}, p_1=\texttt{SwitchLeft}\\
        \mathsf{M}_{10} ,\text{if~} l=\mathsf{T}_1, p_0=\texttt{Cruise}, p_1=\texttt{SwitchRight}
    \end{cases}
    \end{equation}
    where $l$ is the current track mode, $p_0$ is the current internal mode, and $p_1$ is the next internal mode.
    This function encodes the "rules of the road" for this map. The first two cases handle the transition to a cruising state: if an agent's next internal intent is \texttt{Cruise}, it is mapped to the appropriate lane-following track mode ($T_0$ or $T_1$). This applies both to an agent already in the lane and one just completing a lane change. The other two cases define how lane changes are initiated: an agent cruising in a lane that selects a valid switch maneuver (\texttt{SwitchLeft} from Lane 1 or \texttt{SwitchRight} from Lane 2) is mapped to the corresponding transitional track mode ($M_{01}$ or $M_{10}$). For all other input combinations, the function is undefined, meaning the transition is not permitted. For instance, attempting to \texttt{SwitchRight} from Lane 1 ($T_0$) has no effect, as no such path exists in the map.
\end{itemize}

\subsection{Agent State  Variables }
\label{sec:variables}
Agents in the scenario are indexed by a finite set $\I$. An individual {\em agent}, $\agent_i$ for $i\in \I$, is modeled as a state machine whose state is defined by the values of two distinct types of variables. The first are {\em discrete variables}, also called {\em modes}, which take values from a finite set and change their state only during discrete transitions. The second are {\em continuous state variables}, which are real-valued and evolve continuously over time. A complete valuation of these variables constitutes the full state of an agent. To illustrate these concepts, we will continue our running example in previous section by defining the state variables for a car agent operating on the two-lane highway map. This model captures the behavior of an agent cruising on a straight road that switches to the left lane to avoid a vehicle ahead.

The discrete state, or mode, of an agent $\agent_i$ is defined by the valuations of two key sets of variables.
The first is its set of {\em internal mode} variables, denoted $\strategy_i$. A valuation of these variables represents the agent's private behavioral intent. The second is its set of {\em track mode} variables, denoted $\maplabel_i$. A valuation of these variables represents the agent's assigned track mode in the context of the map (e.g., $T_0$ or $M_{01}$), and the values are drawn from the map's global set of track modes, $\maplabel_{\map}$.

For the agent $\agent_i$ in our example shown in Figure~\ref{fig:passing_simple}, the set of {\em internal mode} variables is the singleton set $\strategy_i=\{p_i\}$ with value space $\val(\strategy_i) = \{\texttt{Cruise}, \texttt{SwitchLeft}, \texttt{SwitchRight}\}$. Similarly, the set of {\em track mode} variable, $\maplabel_i$ must be drawn from the set of track modes $\val(\maplabel_i) = \{\mathsf{T_0}, \mathsf{T_1}, \mathsf{M_{01}}, \mathsf{M_{10}}\}$ provided in our Highway example.

The collection of all discrete variables for agent $i$ is therefore the union of the internal and track mode as $\Dagent_i = \strategy_i \cup \maplabel_i$. Consequently, the agent's discrete state space—the set of all possible valuations for these variables—is given by the Cartesian product of the respective value spaces: $\val(\maplabel_i)\times\val(\strategy_i)$. These formalisms are essential for describing the compatibility between an agent and the map, a topic detailed in Section~\ref{sec:scenario}.

The continuous state of agent $i$ is determined by the valuations of a set of real-valued variables, denoted $\Xagent_i$. This set is itself the union of two groups of variables: $\Xagent_i = \workspace_i \cup \Zagent_i$. Here, $\workspace_i$ is the set of variables describing the agent's state in relation to the workspace and map's structure (typically including its physical coordinates), while $\Zagent_i$ contains all other continuous variables that pertain to the agent's own internal dynamics, such as its velocity or internal timers.

In our example, the set of {\em workspace variables} is $\workspace_i = \{x, y\}$, which represent the agent's real-valued 2D position. The remaining continuous variables are its speed and orientation, given by the set $\Zagent_i = \{v, \theta\}$.
The complete set of continuous state variables, $\Xagent_i$, is the union of these two sets, explicitly defined as $\Xagent_i = \workspace_i \cup \Zagent_i = \{x, y, v, \theta\}$. A specific continuous state for agent $i$ is therefore a valuation of these four variables.

The full state of the agent is then defined by the combination of its continuous and discrete variables. The complete set of all state variables for agent $i$ is the union $\agentfullstate_i = \Xagent_i \cup \Dagent_i$, where $\Dagent_i$ is the set of discrete variables defined previously. Consequently, the agent's full state space, the set of all possible valuations, can be expressed as the Cartesian product of the continuous and discrete state spaces: $\val(\Xagent_i) \times \val(\Dagent_i)$.

\subsection{Input and Observable Variables}
\label{sec:io_variables}
In \ourtool, agents interact with each other in a multi-agent scenario via input and observable variables.

\paragraph{Observable Variables}
In \ourtool, each agent exposes a subset of its state variables for other agents to observe. We refer to these as an agent's {\em observable variables}. The specific set of variables that an agent $\agent_i$ is permitted to observe from another agent $\agent_j$ is denoted by $\agentobs_{i,j}$. The contents of this set can depend on both the observing agent $i$ and the agent being observed $j$.

For convenience, we define $\agentobs_i$ as the complete set of all variables that agent $\agent_i$ can observe from every agent in the scenario, including itself. This set is formally the union of all pairwise observable sets:
\[ \agentobs_i = \bigcup_{j \in \I} \agentobs_{i,j} \]
The mechanism that determines which variables are observable is defined by the scenario's sensor model, which is detailed in Section~\ref{sec:sensor}.

\paragraph{Input Variables}
To compute its next state, an agent $\agent_i$ relies on a set of {\em input variables}. This set, denoted by $\agentinvariables_i$, comprises all the variables that agent $\agent_i$ uses for its internal decision logic.

It is crucial to distinguish these input variables from the observable variables ($\agentobs_i$). While the two sets are not necessarily identical, a fundamental requirement for any valid scenario is that an agent can only use information it can actually observe. Therefore, the set of an agent's input variables must be a subset of its observable variables. This constraint is formally expressed as: $\agentinvariables_i \subseteq \agentobs_i$

\begin{figure}[h!]
    \centering
    \scriptsize
\begin{lstlisting}[xleftmargin=20pt]
pass_scenario = Scenario(ScenarioConfig())
Car1 = CarAgent("car1", "passing_simple.py",
initial_state=[[0, -0.8, 0, 1.0], [0.1, 0.8, 0, 1.0]],
initial_mode=(AgentMode.Normal, TrackMode.T1),)
Car2 = CarAgent("car1" ....
pass_scenario.add_agent(Car1)
pass_scenario.add_agent(Car2)
...
//Choose one of the following
scenario.set_sensor(BaseSensor())   // Scenario with Omniscient Agent
scenario.set_sensor(DistSensor())   // Scenario with Distance Agent
pass_scenario.set_map(TwoLaneHighWayMap())
traces = scenario.simulate(20, 0.01)
traces = scenario.verify(20, 0.01)
\end{lstlisting}
\caption{Example \ourtool script for setting up the two-car highway scenario. Note the selection between different observation models for $agent_1$ (lines 10-11).}
\label{fig:passing_scenario}
\end{figure}

To illustrate these concepts, let's further expand our two-lane highway example by adding two car agents, $\agent_1$ (Car1) and $\agent_2$ (Car2) as shown in Figure~\ref{fig:passing_scenario}.
We will analyze two different sensor and decision logic configuration for $\agent_1$
\paragraph{Omniscient Agent.} First, let $\agent_1$ use the decision logic from Figure~\ref{fig:passing_simple}. This agent's logic is designed to use the full state of both agents. Therefore, its required set of input variables is the union of all state variables from both agents: $\agentinvariables_1 = \agentfullstate_1 \cup \agentfullstate_2$. We assume $\agent_1$ has omniscient observation, meaning its set of observable variables is the full states of both agents: $\agentobs_1 = \agentobs_{1,1} \cup \agentobs_{1,2} = \agentfullstate_1 \cup \agentfullstate_2$. Since $\agentinvariables_1 \subseteq \agentobs_1$ is satisfied, the configuration is \rev{valid}. 
\paragraph{Distance Agent.} Now, consider a different version of $\agent_1$ that uses the simpler decision logic from Figure~\ref{fig:simplierDL}, that only requires its own state and the relative distance to $\agent_2$, a variable we'll call \texttt{dist}. Its input variables are now $\agentinvariables_1 = \agentfullstate_1 \cup \{\texttt{dist}\}$. For this to be valid, the relative distance must be explicitly provided to the agent. Therefore, the set of observable variables $\agentobs_1$ must be defined to contain $\agentfullstate_1$ and \texttt{dist}, again satisfying the compatibility check $\agentinvariables_1 \subseteq \agentobs_1$. 

\begin{figure}
    \centering
    \scriptsize
\begin{lstlisting}[xleftmargin=20pt]
class agent_mode: {Cruise, SwitchLeft, SwitchRight}
class TrackMode: {T0, M01, M10, T1}
class State:
    x,y,v,theta : float
    agent_mode : agent_mode
    track_mode : TrackMode

def decisionLogic(ego: State, other: State, map):
    s_next = s
    if ego.agent_mode == Cruise and ego.track_mode == other.track_mode and 6 < other.x-ego.x < 8:
            #start changing to left lane
            if map.track_exist(ego.track_mode, ego.agent_mode, SwitchLeft):
                s_next.agent_mode = SwitchLeft
                s_next.track_mode = map.track(ego.track_mode, ego.agent_mode, SwitchLeft)
            #change to right lane
            ...
    #once lane changes to left lane, return to Cruise
    if ego.agent_mode == SwitchLeft:
        if ego.y >= 2.5:
            s_next.agent_mode = Cruise
            s_next.track_mode = map.track(ego.track_mode, ego.agent_mode, Cruise)
    #handle right lane
    ...
    return s_next
\end{lstlisting}
\caption{Example decision logic for the Omniscient Agent. The function is intended to have access to full states of both ego and other agents.}
\label{fig:passing_simple}
\end{figure}

\begin{figure}
    \centering
    \scriptsize
\begin{lstlisting}[xleftmargin=20pt]
class AgentMode: {Cruise, SwitchLeft, SwitchRight}
class TrackMode: {T0, T1, M10, M01}

class State:
    x,y,v,theta: float
    agent_mode: AgentMode
    track_mode: TrackMode

def decisionLogic(ego: State, dist: float, map):
    #...
\end{lstlisting}
\caption{Simplified decision logic for the Distance Agent. The function signature shows its reliance on a pre-computed dist variable, rather than the full environment.}
\label{fig:simplierDL}
\end{figure}

The set of observable variables $\agentobs_i$ serves as the crucial interface between the global state of the scenario and an agent's internal logic. The component responsible for producing these observations is the sensor model, which is the focus of the next section.

\subsection{Sensors}
\label{sec:sensor}
As mentioned in Section~\ref{sec:io_variables}, a sensor is a module that provides an agent with its observable information, which is computed from the state of all agents in the scenario. Each agent $\agent_i$ has its own specific sensor model, denoted $\sensor_i$.

Formally, the sensor for agent $\agent_i$ is a function that maps the complete state of the scenario to a set of possible valuations for that agent's observable variables. The sensor function has the following signature:
\[ \sensor_i: \val(\agentfullstate_1) \times \ldots \times \val(\agentfullstate_n) \rightarrow 2^{\val(\agentobs_i)} \]

The key feature of this design is that the sensor's output is a set of valuations, not necessarily just a single one. This allows the sensor model to represent uncertainty directly. For example, a sensor determining that the distance to another car is within a range (e.g., [10,12] meters) would return the set of all valuations where the distance variable falls within that interval. 

\subsection{The Agent}

We define $\I$ as the finite index set of agents. 
An {\em agent} $\agent_i$ is formally defined by a tuple containing its state variables, initial conditions, and dynamics. The components of this tuple define the agent's state, its perception of the world, and its behavior over time $\agent_i=\langle \Xagent_i, \Xagent_i^0, \Dagent_i, \Dagent_i^0, \agentinvariables_i, \agentobs_i, \agentguard_i, \agentreset_i, \agentflow_i\rangle$. We will now define each of these components.

\paragraph{State Space and Initialization}
The agent's state is defined by its set of state variables, $\agentfullstate_i = \Xagent_i \cup \Dagent_i$. The agent's state space is the set of all possible valuations of these variables, denoted $\val(\agentfullstate_i)$. The set of possible initial states for the agent is a subset of this space, denoted $\agentfullstate_i^0 \subseteq \val(\agentfullstate_i)$.

\paragraph{Perception and Inputs}
The agent's perception is governed by its sensor model, $\sensor_i$. As defined in Section~\ref{sec:sensor}, this function takes the full system state and produces a set of valuations for the agent's observable variables, $\agentobs_i$. The agent's internal decision logic, in turn, acts upon its set of input variables, $\agentinvariables_i$, which must be a subset of its observable variables ($\agentinvariables_i \subseteq \agentobs_i$).

\paragraph{Discrete Transitions (Guards and Resets)}
The agent's discrete behavior consists of instantaneous transitions between its discrete modes. Each potential transition is governed by a guard and a reset function.
\begin{itemize}
\item The Guard $\agentguard_i$ defines the conditions under which a transition can occur. For each pair of discrete modes $(\dagent, \dagent')$, the guard $\agentguard_i(\dagent, \dagent')$ specifies a subset of the input value space, $\val(\agentinvariables_i)$. A transition from $\dagent$ to $\dagent'$ is enabled if the agent's current input valuation falls within this set.
\item The Reset $\agentreset_i$ defines how the agent's continuous variables are updated during a discrete transition. For each transition $(\dagent,\dagent')$, the reset is a function that maps the agent's pre-transition state to its new continuous state: $\agentreset_i(\dagent, \dagent'): \val(\agentinvariables_i) \rightarrow \val(\Xagent_i)$.
\end{itemize}

\paragraph{Continuous Evolution (Flow)}
While the agent remains in a fixed discrete mode $d$, its continuous state variables $\Xagent_i$ evolve over time. This evolution is described by the agent's {\em flow function}, $\agentflow_i$, which determines the resulting continuous state after a given duration of evolution.

Formally, the function takes an initial continuous state, the agent's current discrete mode, and a time duration, and it returns the single continuous state at the end of that duration. The signature is:
\[ \agentflow_i: \val(\Xagent_i) \times \val(\Dagent_i) \times \mathbb{R}_{\ge 0} \rightarrow \val(\Xagent_i) \]
The underlying dynamics that this function represents may be specified by a set of ordinary differential equations (ODEs) or, more generally, be given by a black-box simulator.

\subsection{Assembling a Complete Scenario}
\label{sec:scenario}
In the preceding sections, we have defined the essential building blocks of our multi-agent scenario: the static map that provides the environment, the agents with their internal state and behavioral rules (guards, resets, and flows), and the sensors that provide the perceptual link between them. We now formalize how these components are assembled into a complete and coherent scenario.

A {\em scenario} in \ourtool is created by combining a map with a collection of agents, as illustrated in Figure \ref{fig:connect}. Formally, a scenario is a tuple $\scenario = \langle \map, \langle \agent_i, \sensor_i\rangle_{i \in \I} \rangle$, consisting of a map $\map$ and an indexed family of $n$ agents coupled with their corresponding sensor.

\begin{figure}
    \centering
    \resizebox{\linewidth}{!}{\begin{tikzpicture}[x=1cm,y=1cm,>=Latex, line cap=round, line join=round]
  \definecolor{scenarioGray}{RGB}{127,127,127}
  \definecolor{scenarioLightGray}{RGB}{166,166,166}
  \definecolor{scenarioTeal}{RGB}{4,36,51}
  \definecolor{scenarioBlue}{RGB}{0,176,240}

  \def\W{18.0}
  \def\H{7.4}
  \def\outerPad{0.3}

  \draw[scenarioGray, line width=1.2pt, rounded corners=0.7cm]
    (\outerPad,\outerPad) rectangle (\W-\outerPad,\H-\outerPad);
  \node[anchor=north west, text=black] at (\outerPad+0.3,\H-\outerPad-0.3) {Scenario};

  \def\mapW{3.6}
  \def\mapH{0.9}
  \def\mapX{9.0}
  \def\mapY{6.3}
  \draw[scenarioLightGray, line width=1.2pt, rounded corners=0.3cm]
    (\mapX-\mapW/2,\mapY-\mapH/2) rectangle (\mapX+\mapW/2,\mapY+\mapH/2);

  \node[text=black] at (\mapX,\mapY) {Map};

  \def\gW{4.8}
  \def\gH{4.2}
  \def\gBottom{1.0}
  \pgfmathsetmacro{\gTop}{\gBottom+\gH}
  \def\topBoxW{4.1}
  \def\topBoxH{1.4}
  \def\topBoxBottom{3.4}
  \def\bottomBoxW{2.6}
  \def\bottomBoxH{1.0}
  \def\bottomBoxBottom{1.5}
  \def\loopOffset{0.2}

  \pgfmathsetmacro{\bracketY}{\gTop+0.35}
  \pgfmathsetmacro{\arrowEndY}{\gTop-0.10}
  \draw[scenarioGray, line width=1.2pt] (\mapX,\mapY-\mapH/2) -- (\mapX,\bracketY);
  \draw[scenarioGray, dashed, line width=1.2pt] (4.0,\bracketY) -- (14.0,\bracketY);
  \foreach \x in {4.0,9.0,14.0} {
    \draw[scenarioGray, -{Latex[length=3mm]}, line width=1.2pt]
      (\x,\bracketY) -- (\x,\arrowEndY);
  }

  \newcommand{\agentgroup}[2]{%
    \pgfmathsetmacro{\gx}{#1}
    \draw[scenarioGray, dashed, line width=1.1pt, rounded corners=0.4cm]
      (\gx-\gW/2,\gBottom) rectangle (\gx+\gW/2,\gTop);

    \draw[laneyellow, line width=1.2pt, rounded corners=0.3cm, fill=white]
      (\gx-\topBoxW/2,\topBoxBottom) rectangle
      (\gx+\topBoxW/2,\topBoxBottom+\topBoxH);

    \def\decW{2.4}
    \def\decH{1.05}
    \draw[scenarioBlue, line width=1.2pt, rounded corners=0.2cm, fill=white]
      (\gx-\decW/2,\topBoxBottom+0.12) rectangle
      (\gx+\decW/2,\topBoxBottom+0.12+\decH);
    \node[text=black] at (\gx,\topBoxBottom+0.12+\decH/2) {Decision Logic};
    \node[text=black] at (\gx-\topBoxW/2+0.55,\topBoxBottom+0.7) {$I_{#2}$};
    \node[text=black] at (\gx+\topBoxW/2-0.55,\topBoxBottom+0.7) {$A_{#2}$};

    \draw[laneyellow, line width=1.2pt, rounded corners=0.3cm, fill=white]
      (\gx-\bottomBoxW/2,\bottomBoxBottom) rectangle
      (\gx+\bottomBoxW/2,\bottomBoxBottom+\bottomBoxH);
    \node[text=black] at (\gx,\bottomBoxBottom+0.5) {$S_{#2}$};

    \draw[scenarioGray, -{Latex[length=3mm]}, line width=1.2pt]
      (\gx-0.7,\gBottom-0.2) -- (\gx-0.7,\bottomBoxBottom);
    \draw[scenarioGray, -{Latex[length=3mm]}, line width=1.2pt]
      (\gx+0.7,\gBottom-0.2) -- (\gx+0.7,\bottomBoxBottom);
    \node[text=black, anchor=north] at (\gx,\gBottom+0.04) {$Y_1 \ldots Y_n$};

    \draw[scenarioGray, line width=1.2pt]
      (\gx+\topBoxW/2,\topBoxBottom+0.55) --
      (\gx+\topBoxW/2+\loopOffset,\topBoxBottom+0.55) --
      (\gx+\topBoxW/2+\loopOffset,\gBottom-0.1);
    \draw[scenarioGray, -{Latex[length=3mm]}, line width=1.2pt]
      (\gx+\topBoxW/2+\loopOffset,\bottomBoxBottom+0.5) --
      (\gx+\bottomBoxW/2,\bottomBoxBottom+0.5);
    \draw[scenarioGray, -{Latex[length=3mm]}, line width=1.2pt]
      (\gx+\topBoxW/2+\loopOffset,\gBottom-0.1) -- (\gx+\topBoxW/2+\loopOffset,0.4);

    \draw[scenarioGray, line width=1.2pt]
      (\gx-\bottomBoxW/2,\bottomBoxBottom+0.5) --
      (\gx-\topBoxW/2-\loopOffset,\bottomBoxBottom+0.5) --
      (\gx-\topBoxW/2-\loopOffset,\topBoxBottom+0.55);
    \draw[scenarioGray, -{Latex[length=3mm]}, line width=1.2pt]
      (\gx-\topBoxW/2-\loopOffset,\topBoxBottom+0.55) --
      (\gx-\topBoxW/2,\topBoxBottom+0.55);

    \pgfmathsetmacro{\leftLoopX}{\gx-\topBoxW/2-\loopOffset}
    \pgfmathsetmacro{\rightLoopX}{\gx+\topBoxW/2+\loopOffset}
    \pgfmathsetmacro{\loopMidY}{(\topBoxBottom+0.55+\bottomBoxBottom+0.5)/2}
    \node[text=black, anchor=west] at (\leftLoopX+0.15,\loopMidY) {$O_{#2}$};
    \node[text=black, anchor=west] at (\rightLoopX-0.65,\loopMidY) {$Y_{#2}$};

  }

  \agentgroup{4.0}{1}
  \agentgroup{9.0}{i}
  \agentgroup{14.0}{n}
\end{tikzpicture}}
    \caption{Architectural overview of a scenario. The sensor $\sensor_i$ processes the Global System State to produce the set of Observable Variables ($\agentobs_i$). The agent $\agent_i$ receives these observations, and its internal Decision Logic uses a subset of this information, known as its Input Variables ($\agentinvariables_i$), to determine its next action. The Map provides the environmental context for all agents.}
    \label{fig:connect}
\end{figure}
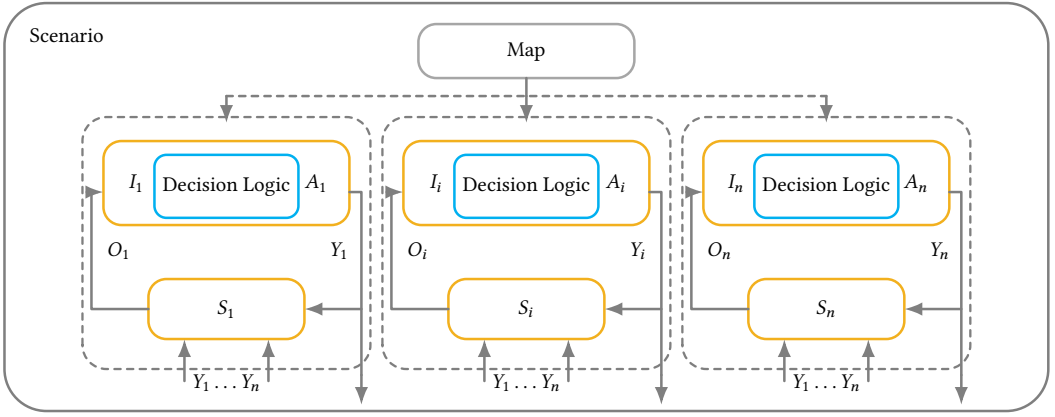

While this tuple combines the necessary components, a scenario is only considered valid if its parts are consistent. These compatibility rules ensure, for example, that an agent's decision-making logic only uses information its sensors can actually provide, and that its intended behaviors are permitted by the map. We formalize these requirements as two key compatibility criteria:

The first criterion governs the flow of information between agents. As established in Section~\ref{sec:io_variables}, an agent's decision logic can only use information that is actually available to it.

\begin{definition}[I/O Compatibility]
An agent $\agent_i$ is I/O compatible within a scenario if its set of required input variables is a subset of its set of observable variables:
\[ \agentinvariables_i \subseteq \agentobs_i \]
\end{definition}

The second criterion ensures that an agent's capabilities and behaviors align with the structure defined by the map.

\begin{definition}[Map Compatibility]
An agent $\agent_i$ is compatible with a map $\map$ if the following conditions hold:
\begin{itemize}
    \item Shared Workspace: The agent's workspace variables and the map's workspace variables must define the same space, i.e., $\val(\workspace_i) = \val(\workspace)$.
    \item Permitted Tactics: Every internal mode the agent can take is permitted by the map, i.e., $\val(\strategy_i) \subseteq \strategy_\map$, and every track mode it can take is one the map provides, i.e., $\val(\maplabel_i) \subseteq \maplabel_\map$.
    \item {\em Admissible transitions:} for every pair of agent modes $\dagent = \langle \ell, p\rangle$ and $\dagent' = \langle \ell', p'\rangle$ with $\agentguard_i(\dagent, \dagent') \neq \emptyset$, the map permits the change and assigns the track mode the agent claims:
    \[ \langle \ell, p, p'\rangle \in \mathrm{dom}(\maplablab) \quad\text{and}\quad \ell' = \maplablab(\ell, p, p'). \]
\end{itemize}
\end{definition}
The third condition is what ties an agent's decision logic to the map it is instantiated on. An agent is written without reference to any particular map, so nothing in its own definition prevents it from proposing a lane change that the map does not offer; the condition rules out exactly those scenarios, and Section~\ref{sec:scenariotoHA} shows that it is also what makes the guard of the composed automaton well defined.
While the scenario model is now complete, it is not yet in a form suitable for rigorous analysis. The next section provides the necessary formalization by mapping the scenario to a hybrid automaton.

\section{\ourtool Scenario to Hybrid Verification}
\label{sec:scenariotoHA}

The scenario defined in Section~\ref{sec:scenario} provides a high-level description of a multi-agent system. To enable formal analysis, we now define the mapping from a \ourtool scenario, $\scenario$, to its underlying mathematical representation as a hybrid automaton, which we denote $\ha(\scenario)$. This formalization is crucial as it allows us to apply rigorous verification algorithms developed for hybrid systems to analyze the safety and behavior of the scenario.

Given a scenario $\scenario = \langle \map, \langle \agent_i, \sensor_i\rangle_{i \in \I} \rangle$, composed of a map and an indexed family of agents, its corresponding hybrid automaton is defined by the tuple:
\[ \ha(\scenario) = \langle \Xha, \Xha^0, \Dha,  \Dha^0, \haguard, \hareset, \tl \rangle \]
where the components representing the state space, initial conditions, and dynamics are defined in the following sections.

\subsection{State Space and Initialization}
The state space of the automaton is composed of continuous and discrete components.

\paragraph{Continuous State Space} The {\em continuous state space} of the automaton is the set of all possible valuations of the system's continuous variables:
\[ \Xha := \val\left(\bigcup_{i \in \I} \Xagent_i\right) \]
A valuation $x \in \Xha$ represents a complete continuous state of the system. The set of {\em initial continuous states}, $\Xha^0 \subseteq \Xha$, is composed of all continuous states that are part of a valid initial configuration of the scenario.

\paragraph{Discrete State Space} The {\em discrete state space}, or {\em mode space}, is the set of all possible valuations of the system's discrete variables:
\[ \Dha := \val\left(\bigcup_{i \in \I} \Dagent_i\right) \]
A valuation $d \in \Dha$ represents a complete discrete state, or mode, of the system. The set of {\em initial modes}, $\Dha^0 \subseteq \Dha$, is composed of all discrete modes that are part of a valid initial scenario configuration.

\subsection{Automaton Dynamics}
An agent's dynamic behavior is defined by its discrete transitions and continuous evolution. Discrete transitions are given an {\em interleaving\/} semantics: a transition of the composed automaton advances exactly one agent, and agents whose guards are enabled at the same instant are represented by the sequences of single-agent transitions that can be taken at that instant, with no time elapsing between them. Proposition~\ref{prop:interleaving} states the condition under which this loses nothing relative to a semantics in which such agents move together.
The components governing these dynamics are defined as follows.

\paragraph{Guard} The {\em guard} for the automaton, $\haguard$, determines when a discrete transition between two global modes, $\dha$ and $\dha'$, is allowed. A transition is enabled only if exactly one agent is making a move while all others remain idle.

To define this formally, let a global discrete mode $\dha \in \Dha$ be a tuple of individual agent modes, $\dha = (d_i)_{i \in \I}$. A transition from $\dha = (d_i)_{i \in \I}$ to $\dha' = (d'_i)_{i \in \I}$ is enabled for a continuous state $\xha \in \Xha$ if and only if there is exactly one active agent $i \in \I$ that satisfies the following conditions:

\begin{enumerate}
    \item \textit{Mode Change:} The active agent $i$ changes its discrete mode, i.e., $d_i \neq d'_i$.

    \item \textit{Synchronicity:} All other agents $j \neq i$ remain in their current modes, i.e., $d_j = d'_j$.

    \item \textit{Local Guard Satisfaction:} The active agent's local guard for its transition from $d_i$ to $d'_i$ is satisfied. Writing $\mathbf{y} = \langle \xha, \dha\rangle$ for the full state of the scenario, this means the set of observations produced by its sensor contains at least one valuation permitted by the guard:
    \[ \left(\sensor_i(\mathbf{y}) \lceil \agentinvariables_i\right) \cap \agentguard_i(d_i, d'_i) \neq \emptyset. \]

\begingroup\ifrevisionmarkup\color{blue}\fi
    \item \textit{Map Admissibility:} writing $d_i = \langle \ell_i, p_i\rangle$ and $d'_i = \langle \ell'_i, p'_i\rangle$, the map permits the change and assigns the new track mode:
    \[ \langle \ell_i, p_i, p'_i\rangle \in \mathrm{dom}(\maplablab) \quad\text{and}\quad \ell'_i = \maplablab(\ell_i, p_i, p'_i). \]
\endgroup
\end{enumerate}

\begingroup\ifrevisionmarkup\color{blue}\fi
The set $\haguard(\dha, \dha')$ is therefore the set of all continuous states $\xha$ for which such an active agent exists.
The map constraints enter the semantics in (4): Without it an agent's decision logic could propose any mode change it liked and the composed automaton would admit it, so a car could switch to a lane that does not exist beside it. Condition (4) forces the reachable mode space to the changes the map actually offers. If $\agent_i$ is map compatible then the second half of the condition is implied by the first, since the guard is non-empty only for transitions the map admits.

\endgroup

\paragraph{Reset Map} The {\em reset map}, $\hareset(\dha, \dha')$, is a function $\hareset(\dha, \dha'):\Xha \rightarrow \Xha$ that defines how the system's continuous state $\xha$ is updated during a discrete transition from $\dha$ to $\dha'$.


 \paragraph{Reset Map} The {\em reset map}, $\hareset(\dha, \dha')$, is a function $\hareset(\dha, \dha'):\Xha\rightarrow \Xha$ that defines how the system's continuous state $\xha$ is updated during a discrete transition. The reset's behavior depends on which agent is active in the transition.

\begingroup\ifrevisionmarkup\color{blue}\fi
By the guard condition above, $\hareset(\dha, \dha')$ is defined only when there is exactly one active agent $i$. Because the sensor is set-valued, a single continuous state may admit several observations that enable the transition, and each may reset the agent differently:
\[ \hareset(\dha, \dha')(\xha) = \left\{\, \xha' \;\middle|\; \xagent'_i = \agentreset_i(d_i, d'_i)(\agentinval),\ \ \xagent'_j = \xagent_j \text{ for } j \neq i,\ \ \agentinval \in \left(\sensor_i(\mathbf{y}) \lceil \agentinvariables_i\right) \cap \agentguard_i(d_i, d'_i) \right\}. \]
Only the active agent's continuous variables are written; every other agent's state remains constant.
\endgroup

\paragraph{Continuous Evolution (Flows)} The continuous dynamics of the hybrid automaton are governed by the flow functions of the individual agents. While the system remains in a fixed global discrete mode $\dha$, the global continuous state $\xha$ evolves according to the combined dynamics of all agents running in parallel.

Formally, a function $\Tha: [0, T] \rightarrow \Xha$ is a valid {\em trajectory} for the system starting from a continuous state $\Tha(0) = \xha^0$ in a constant mode $\dha = (d_i)_{i \in \I}$ if, for every agent $i \in \I$ and for all time $t \in [0,T]$, the following condition holds:
\[ \Tha_i(t) = \agentflow_i(x^0_i, d_i, t) \]
Here, $\Tha_i(t)$ denotes the component of the global continuous state corresponding to agent $i$ at time $t$, and $x^0_i$ is the initial continuous state of agent $i$ (i.e., the $i$-th component of $\xha^0$). This definition shows how the global system's trajectory is composed from the individual trajectories defined by each agent's flow function, $\agentflow_i$.
The component $\tl$ is therefore the set of all such pairs $\langle \Tha, \dha\rangle$ representing the valid continuous dynamics of the system.
The underlying dynamics this function represents are defined by a simulator. While this is typically a black-box function, analytical "white-box" models, such as ordinary differential equations (ODEs), are also supported by \ourtool, provided they are executable within its Python environment.

\subsection{Formal Semantics as a Hybrid Automaton}
\label{sec:semantics}
\begingroup\ifrevisionmarkup\color{blue}\fi
The construction of the previous sections maps each component of a scenario to a component of $\ha(\scenario) = \langle \Xha, \Xha^0, \Dha, \Dha^0, \haguard, \hareset, \tl \rangle$. We summarize the components in the table below.

\begin{table}[!htp]
\centering
\label{tab:components}
\small
\begin{tabular}{l l l}
\toprule
Automaton & Determined by & Definition \\
\midrule
$\Xha$   & agents' continuous variables & $\val(\bigcup_{i} \Xagent_i)$ \\
$\Dha$   & agents' modes, internal and track & $\val(\bigcup_{i} \Dagent_i)$, $\Dagent_i = \strategy_i \cup \maplabel_i$ \\
$\Xha^0, \Dha^0$ & agents' initial conditions & projections of $\bigtimes_i \agentfullstate_i^0$ \\
$\haguard$ & guards, sensors, and the map's $\maplablab$ & one active agent; $(\sensor_i \lceil \agentinvariables_i) \cap \agentguard_i \neq \emptyset$; $\maplablab$ admits the change \\
$\hareset$ & the active agent's reset & $\agentreset_i$ on $\Xagent_i$, identity elsewhere \\
$\tl$      & agents' flows & $\Tha_i(t) = \agentflow_i(\xagent^0_i, d_i, t)$ for every $i$ \\
\bottomrule
\end{tabular}
\end{table}

Fix $\delta > 0$. A {\em $\delta$-run\/} of $\scenario$ is a sequence $\langle \sigma^0, \dha^0 \rangle, \ldots, \langle \sigma^{m-1}, \dha^{m-1}\rangle$ in which each $\sigma^k$ assigns to every agent $i$ a function $\sigma^k_i : [0,\delta] \rightarrow \val(\Xagent_i)$, such that
\begin{enumerate}
    \item {\em Initialization.} $\langle \sigma^0_i(0), \dha^0_i \rangle \in \agentfullstate_i^0$ for every agent $i$.
    \item {\em Evolution.} $\sigma^k_i(t) = \agentflow_i(\sigma^k_i(0), \dha^k_i, t)$ for every agent $i$, every $k$, and all $t \in [0, \delta]$.
    \item {\em Interaction.} For each $k < m-1$ there is exactly one agent $i$ with $\dha^{k+1}_i \neq \dha^k_i$; the others satisfy $\dha^{k+1}_j = \dha^k_j$ and $\sigma^{k+1}_j(0) = \sigma^k_j(\delta)$. Writing $\mathbf{y}$ for the state of the scenario at the end of step $k$, there is an observation $\agentinval \in (\sensor_i(\mathbf{y}) \lceil \agentinvariables_i) \cap \agentguard_i(\dha^k_i, \dha^{k+1}_i)$, the map admits the change, and $\sigma^{k+1}_i(0) = \agentreset_i(\dha^k_i, \dha^{k+1}_i)(\agentinval)$.
\end{enumerate}
A $\delta$-run mentions only agents, sensors and the map, and no global mode space, guard or reset.

\begin{theorem}
\label{thm:correspondence}
Let $\scenario = \langle \map, \langle \agent_i, \sensor_i\rangle_{i \in \I}\rangle$ be a scenario with I/O compatible and map compatible agents. Then $\ha(\scenario)$ is a hybrid automaton in the sense of Definition~1.1 of~\cite{561342}, and for every $\delta > 0$ a sequence is a $\delta$-run of $\scenario$ if and only if its componentwise image is a $\delta$-execution of $\ha(\scenario)$. Consequently $\reach{}(\scenario) = \reach{}(\ha(\scenario))$, and $\scenario$ satisfies a safety assertion exactly when $\ha(\scenario)$ does.
\end{theorem}

\begin{proof}
 The invariant of every control mode is $\mathit{true}$ and the event alphabet $\Sigma$ is empty, since transitions are autonomous: agents influence one another only through the observations their guards read.

Let $\Phi$ send a scenario state $\langle (\sigma_i)_{i}, (\dha_i)_{i}\rangle$ to $\langle \xha, \dha \rangle$ by juxtaposing components. It is a bijection onto $\Xha \times \Dha$, which are the valuations of the unions of the agents' variables, and I/O compatibility leaves no variable of a guard unbound. So it suffices that $\Phi$ carry single steps to single steps both ways.

Take a step of a $\delta$-run. Condition~2 is the condition defining $\tl$, so $\langle \Phi(\sigma^k), \dha^k \rangle \in \tl$. Condition~3 supplies an agent $i$ changing mode while the rest stand still, an observation in $(\sensor_i(\mathbf{y}) \lceil \agentinvariables_i) \cap \agentguard_i(\dha^k_i, \dha^{k+1}_i)$, and a map-admissible change: the four clauses of $\haguard$. Hence $\Phi(\sigma^k(\delta)) \in \haguard(\dha^k, \dha^{k+1})$, and the same observation witnesses $\Phi(\sigma^{k+1}(0)) \in \hareset(\dha^k, \dha^{k+1})(\Phi(\sigma^k(\delta)))$.

Conversely, membership in $\tl$ unpacks to Condition~2 agent by agent, and membership in $\haguard(\dha^k, \dha^{k+1})$ yields a unique active agent and an observation satisfying its guard, giving Condition~3. Map compatibility makes the track mode named by $\dha^{k+1}_i$ the one $\maplablab$ assigns, so the step is realizable by $\agent_i$ on $\map$. Induction along the sequence gives the claim, and the reachable sets agree since corresponding runs end at states identified by $\Phi$.
\end{proof}

The compatibility conditions are not  vacuous. Without I/O compatibility---say the decision logic branches on \texttt{other.v} while the sensor exposes only \texttt{dist}, as in Section~\ref{sec:io_variables}---$\agentguard_i$ constrains a variable no observation binds, $\Phi$ is not a bijection, and $\haguard$ is not a well-defined subset of $\Xha$. Without map compatibility the guard admits a pair at which $\maplablab$ is undefined, as the highway map is for \texttt{SwitchRight} from $\mathsf{T_0}$, and the construction admits transitions no agent can take on the map. \ourtool checks both when a scenario is assembled.
\endgroup

\subsection{On the interleaving semantics}
\label{sec:interleaving}
\begingroup\ifrevisionmarkup\color{blue}\fi
A transition of $\ha(\scenario)$ advances one agent, so consider two agents whose guards are enabled by the same observation, as when two cars each detect the other and both decide to change lane. The modeler need not separate observation from action: both orders are behaviors of $\ha(\scenario)$, and the two occurs at the same real time. Interleaving fixes an order, and the following Proposition states when the order cannot matter.

\begin{proposition}
\label{prop:interleaving}
Suppose for every $i \neq j$ that the variables $\agentreset_i$ modifies are disjoint from $\agentobs_{j,i}$. Then for any two agents enabled at the same state, both interleavings reach the same state, and that  is the same state a joint transition modeling a simultaneous move would reach.
\end{proposition}

\begin{proof}
Let $i$ and $j$ both be enabled at $\langle \xha, \dha\rangle$. Taking $i$ first changes $\dha_i$ and writes only variables of $\Xagent_i$, none of which is in $\agentobs_{j,i}$ by hypothesis. So $\sensor_j$ returns the same observations as before, $\agent_j$'s guard is still enabled, its reset evaluates to the same value, and $\dha_j$ is untouched. Symmetrically with the roles exchanged. Both orders write the same values to disjoint variables, and agree with the simultaneous step, which applies both resets to the pre-transition state.
\end{proof}

When some reset does write a variable another agent observes, $\agent_j$ reads a value that depends on whether $\agent_i$ has already moved, and the two transition orderings differ. \ourtool explores both, so analysis stays sound, but its reachable set is then strictly larger than a simultaneous semantics would give. The hypothesis of Proposition~\ref{prop:interleaving} holds throughout Section~\ref{sec:exp}.
\endgroup

\subsection{Sampled-data control}
\label{sec:sampleddata}
\begingroup\ifrevisionmarkup\color{blue}\fi

A transition of $\ha(\scenario)$ changes an agent's mode, so a controller that reads a sensor and holds its output within one mode is not a self-loop. It can be written in two ways.

The sampling can go inside the flow: $\agentflow_i$ is an arbitrary simulator, so a controller of period $h$ is realized by one that recomputes its input at multiples of $h$ and holds it in between, leaving the agent in one mode with no transitions. Soundness is unaffected, since Theorem~\ref{thm:correspondence} requires only that the trajectories of $\tl$ agree with $\agentflow_i$, and $\postcont$ over-approximates whatever the simulator produces.

Alternatively each sample becomes a transition, via a clock $\tau$ with $\dot\tau = 1$ and a mode variable alternating on $\tau = h$, reset to recompute the control and set $\tau := 0$. The sampling instants are then visible to the analysis, at one transition per sample.

The two modeling styles differ in precision. Sampling inside the flow costs accuracy only when the control is discontinuous in the sampled state, since nearby trajectories then receive different inputs and the fitted discrepancy bound must cover the difference; proportional feedback introduces no such discontinuity, quantizing the actuator does. Making each sample a transition avoids this but re-approximates the reachable set at every sample, which dominates over many periods. Section~\ref{sec:exp:sampled} measures both. The flow formulation is preferable unless the sampling instants are themselves of interest.
\endgroup

\subsection{Verification algorithms in \ourtool}
\label{sec:verification_algorithm}
To perform verification, we first formalize a system's behavior over time. We then present the tree-based reachability algorithm implemented in \ourtool that analyzes these behaviors for safety.

\paragraph{Executions and the Safety Verification Problem}
We formalize a system's behavior as a $\delta$-execution, which is a time-stepped sequence of trajectories. For a given trajectory $\Tha$, we use the shorthand $\Tha.fstate$, $\Tha.lstate$, and $\Tha.ltime$ to denote its initial state $\Tha(0)$, final state $\Tha(T)$, and duration T, respectively. A $\delta$-execution for a hybrid automaton $\ha(\scenario)$ is then a sequence of labeled trajectories $\exec := \langle \Tha^0, \dha^0\rangle,...,\langle\Tha^{m-1}, \dha^{m-1}\rangle$.

This sequence must satisfy three conditions. First, the execution must begin from a valid initial state of the automaton, such that $\Tha^0.fstate \in \Xha^0$ and $\dha^0 \in \Dha^0$. Second, each transition in the sequence must satisfy a valid discrete transition. This means that for each trajectory $\Tha^i$, its endpoint $\Tha^i.lstate$ must satisfy the guard condition $\haguard(\dha^i, \dha^{i+1})$ to enable a jump to the next mode. Consequently, the start point of the next trajectory, $\Tha^{i+1}.fstate$, is determined by the reset map $\hareset(\dha^i, \dha^{i+1})(\Tha^i.lstate)$. Third, each trajectory segment has a fixed duration $\Tha^i.ltime = \delta$, except for the final segment, which can be shorter.

From this definition, we can characterize the set of all reachable states. We define first and last state of an execution $\exec = \langle \Tha^0, \dha^0\rangle,...,\langle\Tha^{m-1}, \dha^{m-1}\rangle$ as $\exec.fstate = \Tha^0.fstate$, $\exec.lstate=\Tha^{m-1}.lstate$ and the first and last mode as $\exec.fmode = \dha^0$ and $\exec.lmode = \dha^{m-1}$. The set of reachable states is defined by $\reach{} := \{\exec.lstate \mid \exec\text{ is an execution of }\ha\}$. In addition, we denote the reachable states in a specific mode $\dha \in \Dha$ as $\reach{}(\dha)$ and $\reach{}(T)$ to be the set of reachable states  at time $T$.

The safety verification problem is then to determine if this reachable set ever intersects with a predefined set of unsafe states. Denoting the unsafe states for a given mode $\dha$ as $\unsafe(\dha)$, the system is considered safe if and only if $\reach{}(\dha) \cap \unsafe(\dha) = \emptyset$ for all modes $\dha \in \Dha$.

\paragraph{Reachability Analysis with Post Operators}
To solve this safety problem, \ourtool's \verify function implements a reachability analysis algorithm that constructs an execution tree. The algorithm relies on two fundamental operators that compute the set of states reachable after a discrete jump or a continuous flow. The discrete post-operator, $\post_{\dha,\dha'}$, calculates the result of a transition between modes $\dha$ and $\dha'$, defined as $\post_{\dha,\dha'}(\xha) = \xha'$ iff $\xha \in \haguard(\dha,\dha')$ and $\xha'=\hareset(\dha,\dha')(\xha)$. The continuous post-operator, $\post_{\dha,\delta}$, computes the result of continuous evolution within mode $\dha$ for a duration $\delta$. More specifically, $\post_{\dha,\delta}(\xha) = \xha'$ iff $\forall i\in {1,...,k}$, $\xha_i'=\agentflow_i(\xha_i, \dha_i, \delta)$. These operators are lifted to operate on sets of states in the usual way.

The algorithm begins at the root of the execution tree, $\langle \Xha^0, \dha^0\rangle$, and explores reachable states in a breadth-first manner. An edge is created from a vertex $\langle \Xhaset, \dha \rangle$ to a new vertex $\langle \Xhaset', \dha'\rangle$
if $\langle \Xhaset, \dha \rangle$ to $\langle \Xhaset', \dha'\rangle$, iff $\Xhaset' = \post_{\dha,\dha'}(\post_{\dha,\delta}(\Xhaset))$. 
This follows the "flow-then-jump" sequence that mirrors the formal definition of an execution. 
During the construction of this tree, the safety conditions are checked at each new vertex. In practice, the functions \postcont and \postdisc compute over-approximations of these sets to ensure soundness, using an approach based on the algorithms in~\cite{dryvr}.

\section{Experiments and Use Cases}
\label{sec:exp}
Our experimental evaluation is designed to demonstrate both the unique capabilities and the comparative performance of \ourtool. The evaluation proceeds in two parts. First, we present a series of examples to highlight the novel features of our formalism, particularly its expressive power in modeling complex, multi-agent scenarios. This section serves to qualitatively validate the unique advantages of our approach. Second, we quantitatively benchmark \ourtool against two state-of-the-art hybrid systems verification tools, CORA~\cite{Althoff2015a} and JuliaReach~\cite{JuliaReach19}. This comparison, conducted on a suite of standard benchmarks, establishes the computational efficiency and scalability of our tool relative to existing methods.
\rev{Unless noted otherwise, every measurement reported in this section was taken on an Apple~M2 with 8 cores and 24\,GB of memory running macOS, single-threaded.}

\subsection{Important Features of \ourtool}
To showcase the key features and algorithms of \ourtool, our initial experiments utilize a diverse set of scenarios. These scenarios are constructed from two primary agent models: a 4-dimensional ground vehicle with bicycle dynamics managed by a Stanley controller~\cite{4282788}, and a 6-dimensional drone guided by a neural network controller~\cite{ivanov2019verisig}. Furthermore, each agent can be equipped with one of two distinct decision logics. The first is a reactive collision avoidance (CA) logic that directs an agent to switch tracks upon detecting another agent ahead, while the second is a simpler non-player vehicle (NPV) logic where the agent follows its track at a constant speed without reacting to others. We denote these combinations with a shorthand, such as C-CA for a car with collision avoidance or D-NPV for a drone with non-player logic. These agents operate within a varied collection of six maps, including four 2D environments ($\map1$-$\map4$) and two 3D environments ($\map5$-$\map6$). The 2D maps feature configurations of three and five parallel straight tracks ($\map1$ and $\map2$), a set of three parallel curved tracks ($\map3$), and a complex road network imported from OpenDRIVE ($\map4$). The 3D environments are a pair of stacked straight corridors ($\map5$) and a figure-8 circuit ($\map6$), in which the reference tracks cross over themselves in the vertical dimension.
 
\paragraph{Safety analysis with multiple drones in a 3-d map.}
The first example is a scenario with two drones---D-CA agent (red) and D-NPV agent (blue)---in map $\map5$. The safety assertion requires agents to always separate by at least 1m. 
Fig.~\ref{fig:exp:simple_2} ({\em left}) shows the  computed reachable set, its projection on $x$-position, and on $z$ position. Since the agents are separated in space-time, the scenario is verified safe. These plots are generated using \ourtool's plotting functions. 

\begin{figure}[!htb]
        \centering
        \includegraphics[width=0.24\textwidth]{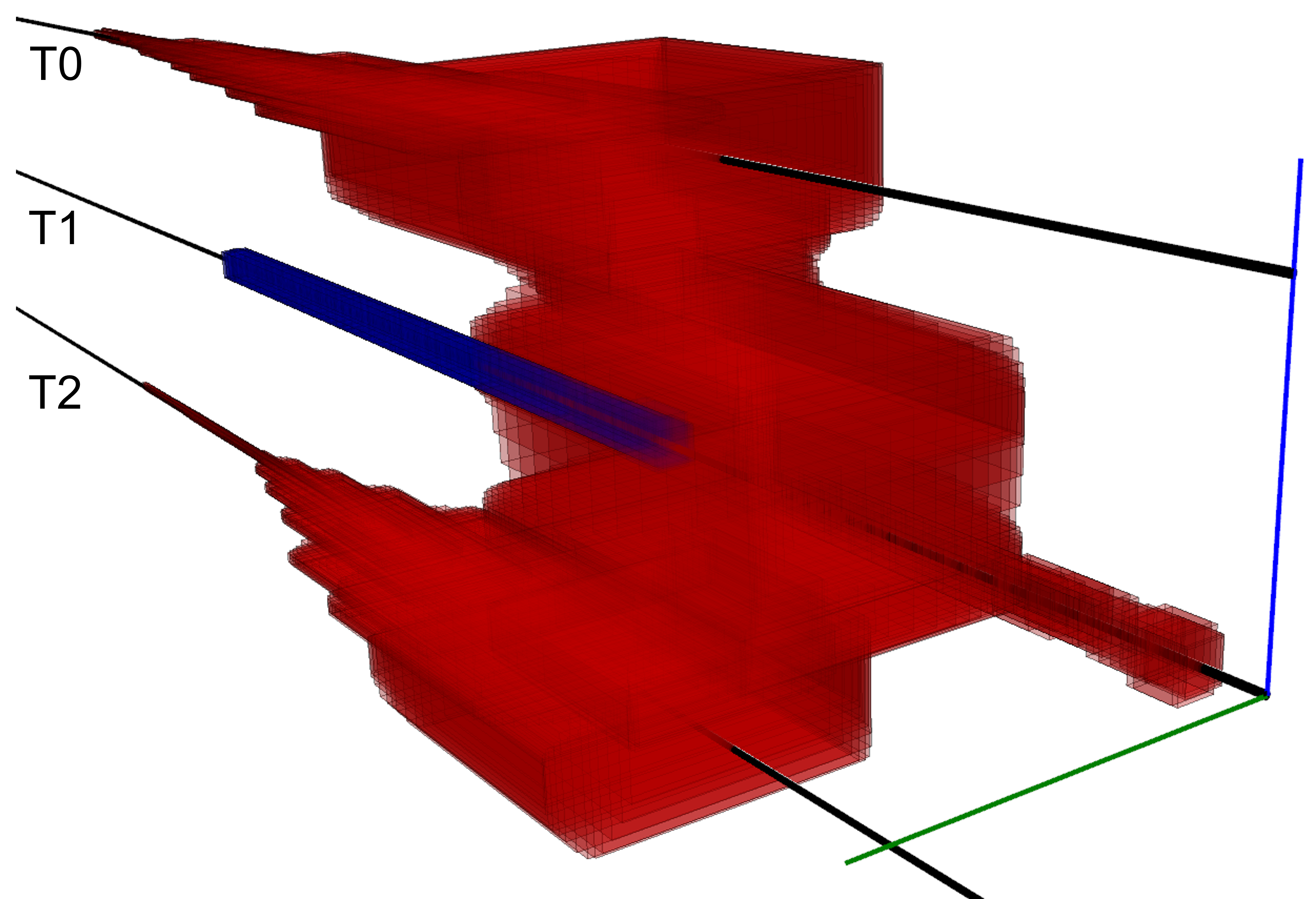}
        \includegraphics[width=0.24\textwidth]{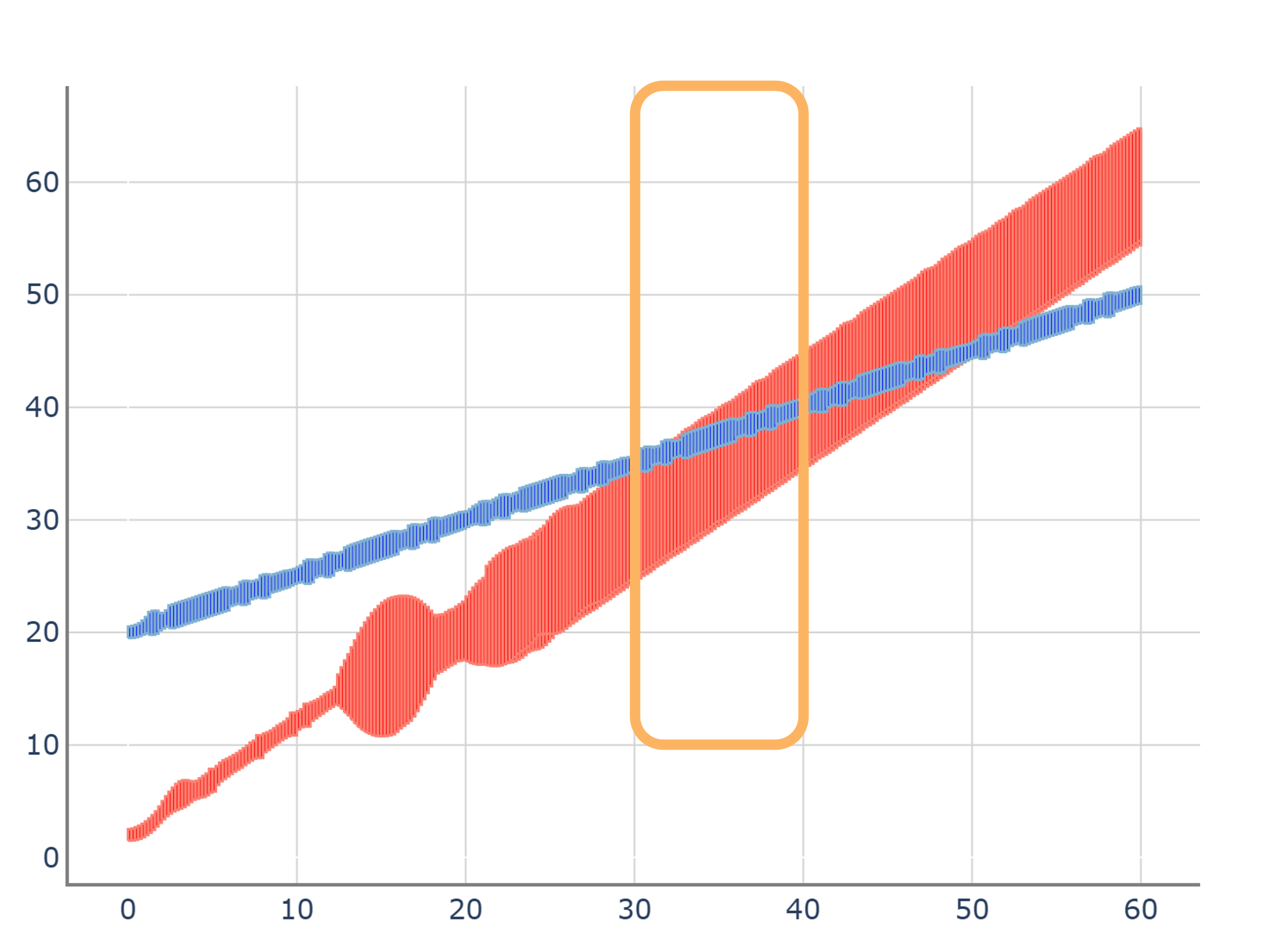}
        \includegraphics[width=0.24\textwidth]{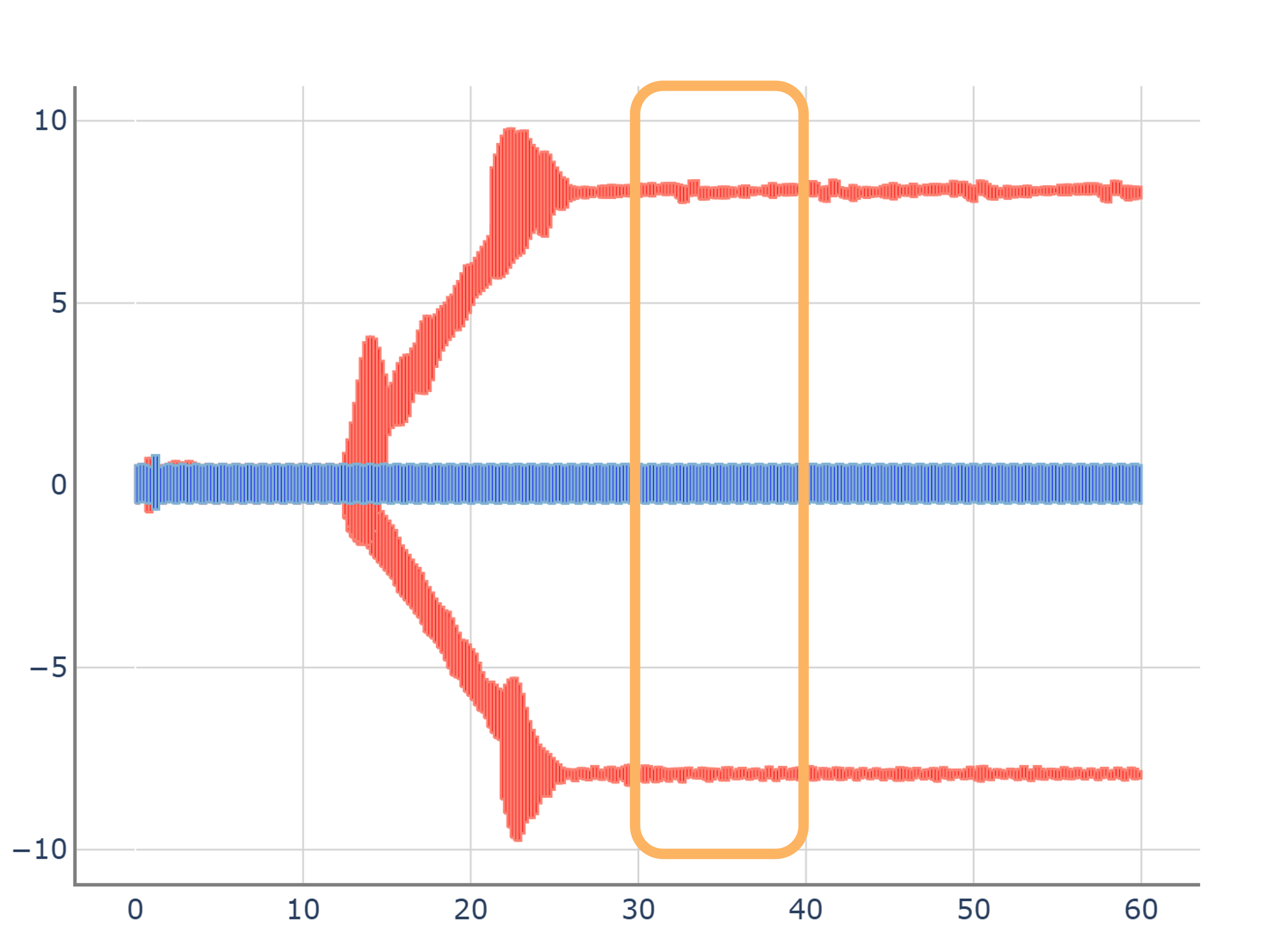}
        \includegraphics[width=0.24\textwidth]{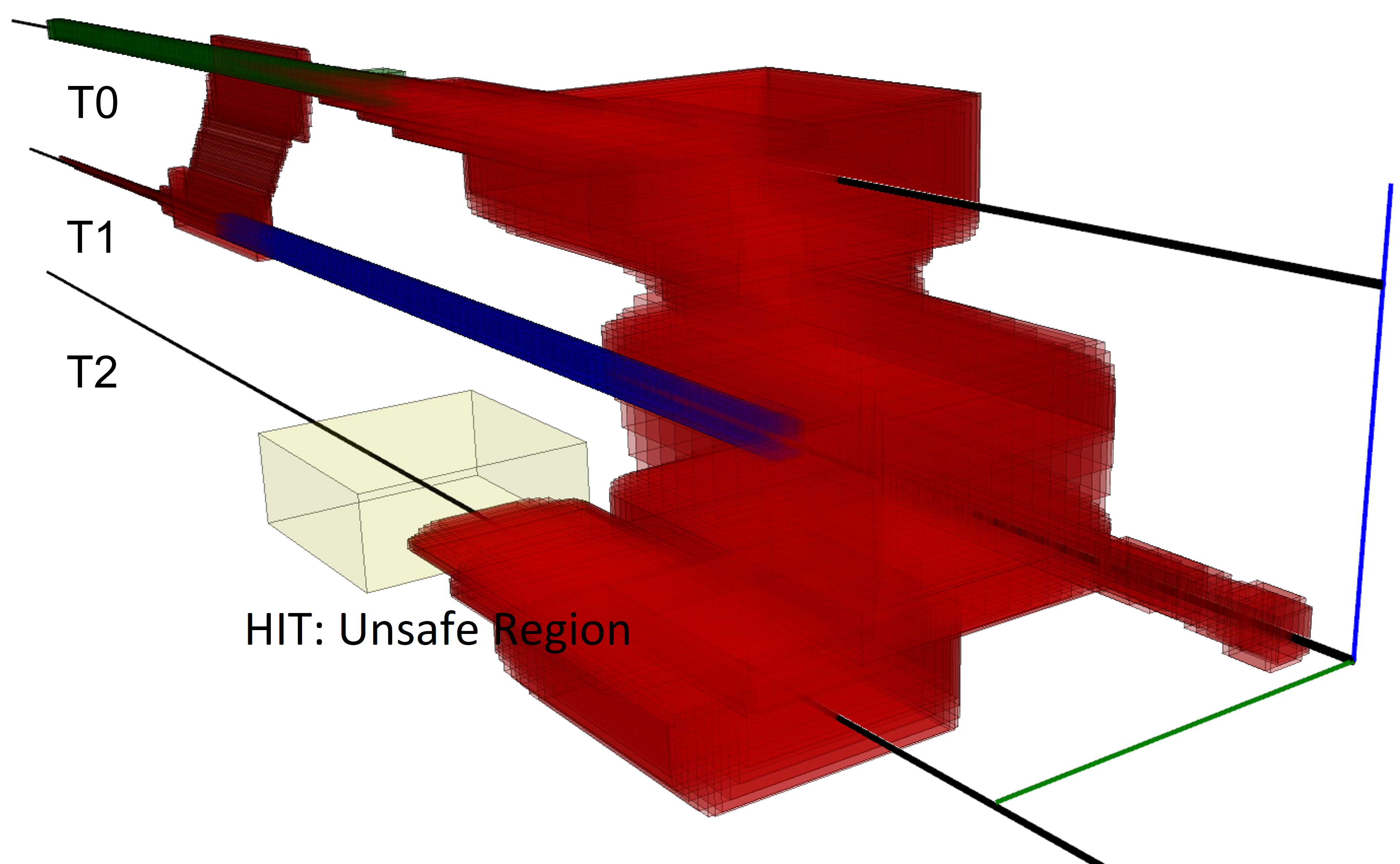}
        \caption{\small Left to right: (1) Computed reachtubes for a 2-drone scenario; (2) same reachtube projected on  x-dimension, and (3) on z-dimension. Since there is no overlap in space-time, no collision. (4) Reachtube for a 3-drone scenario, the red drone violates the safety condition by entering the unsafe region after moving downward.}\label{fig:exp:simple_2}
\end{figure}

\paragraph{Checking multiple safety assertions.} 
\ourtool supports multiple safety assertions specified using $\texttt{assert}$ statements. 
For example, the user can specify unsafe regions (Line 77-78) or safe separation between agents (Line 79-82) as shown in Fig.~\ref{code:unsafe}.
We add a second D-NPV to the previous scenario and both safety assertions. The result is shown in the rightmost Fig.~\ref{fig:exp:simple_2}. In this scenario, D-CA violates the safety property by entering the unsafe region after moving downward to avoid collision.
The behavior of D-CA after moving upward is not influenced. There is no violation of safe separation. \ourtool allow users to extract the set of reachable states and mode transitions that leads to a safety violation. 

\begin{figure}
    \centering
    \scriptsize
\begin{lstlisting}[backgroundcolor=\color{bg},xleftmargin=20pt,
firstline=77, lastline=82, firstnumber=77]
    assert not (ego.x > 40 and ego.x < 50 and \
        ego.y > -5 and ego.y < 5 and ego.z > -10 and ego.z < -6), "Unsafe Region"
    assert not any(ego.x-other.x < 1 and ego.x-other.x > -1 and \
        ego.y-other.y < 1 and ego.y-other.y > -1 and \ 
        ego.z-other.z < 1 and ego.z-other.z > -1 \
        for other in others), "Safe Separation"
\end{lstlisting}
\caption{\small Safety assertions for three drone scenario.}
\label{code:unsafe}
\end{figure}

\paragraph{Changing maps.} \ourtool allows users to easily create scenarios with different maps and port agents across compatible maps. We start with a  scenario with one C-CA agent (red) and two C-NPV agents (blue, green) in $\map1$. The safety assertion  is that the vehicles should be at least 1m apart in both $x$ and $y$-dimensions.  Fig.~\ref{fig:exp:map_safe} ({\em left\/}) shows the  verification result  and safety is not violated. 
However, if we switch to map $\map3$ by changing one line in the scenario definition, a reachability analysis shows that a safety violation can happen after C-CA merges left  Fig.~\ref{fig:exp:map_safe} ({\em center}). 
In addition, \ourtool allows importing map from OpenDRIVE \cite{noauthor_asam_nodate} format. An example is shown in Fig.~\ref{fig:exp:additional} ({\em left}).

\begin{figure}
    \centering
    \includegraphics[width=0.32\textwidth]{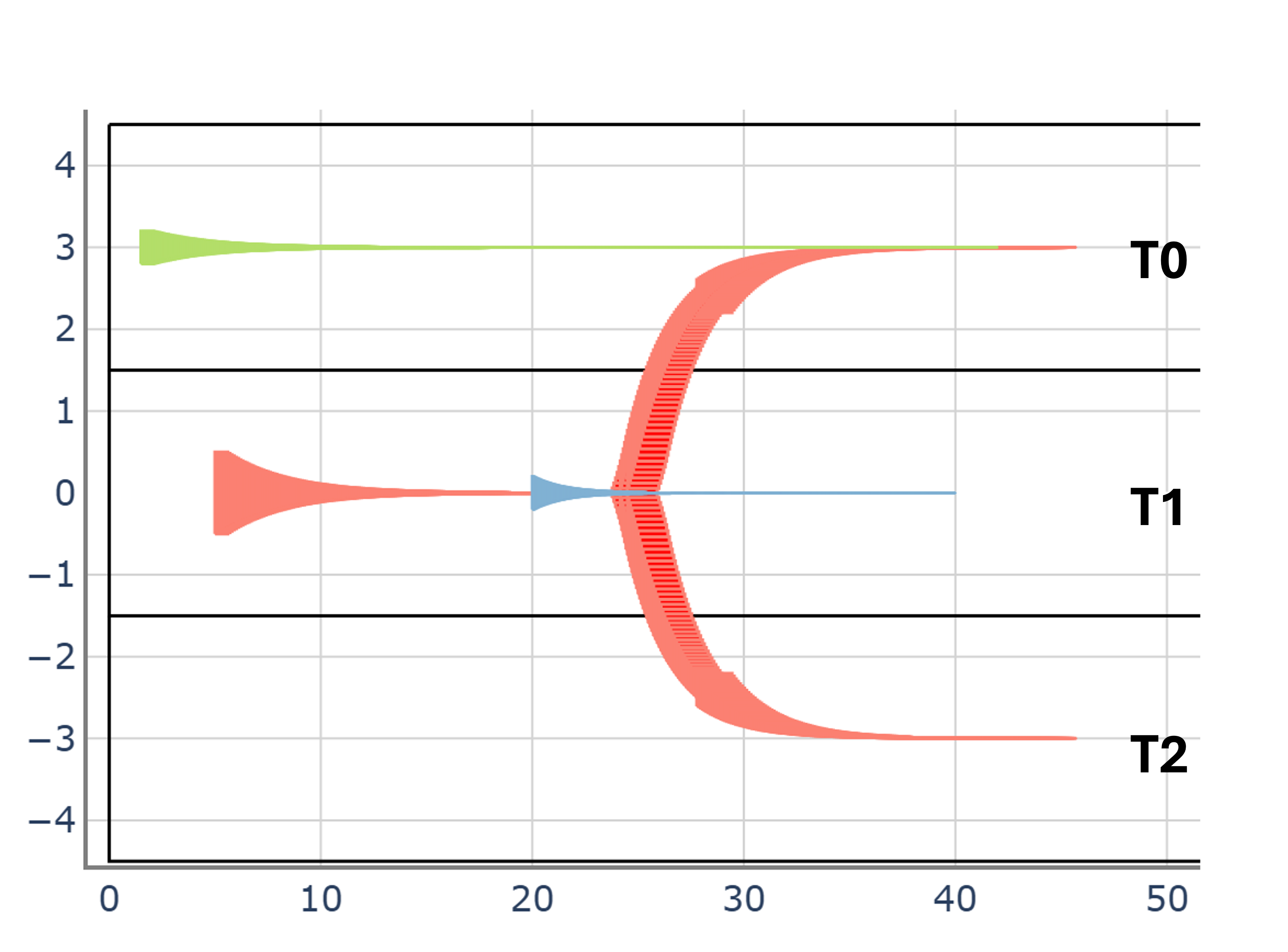}
    \includegraphics[width=0.32\textwidth]{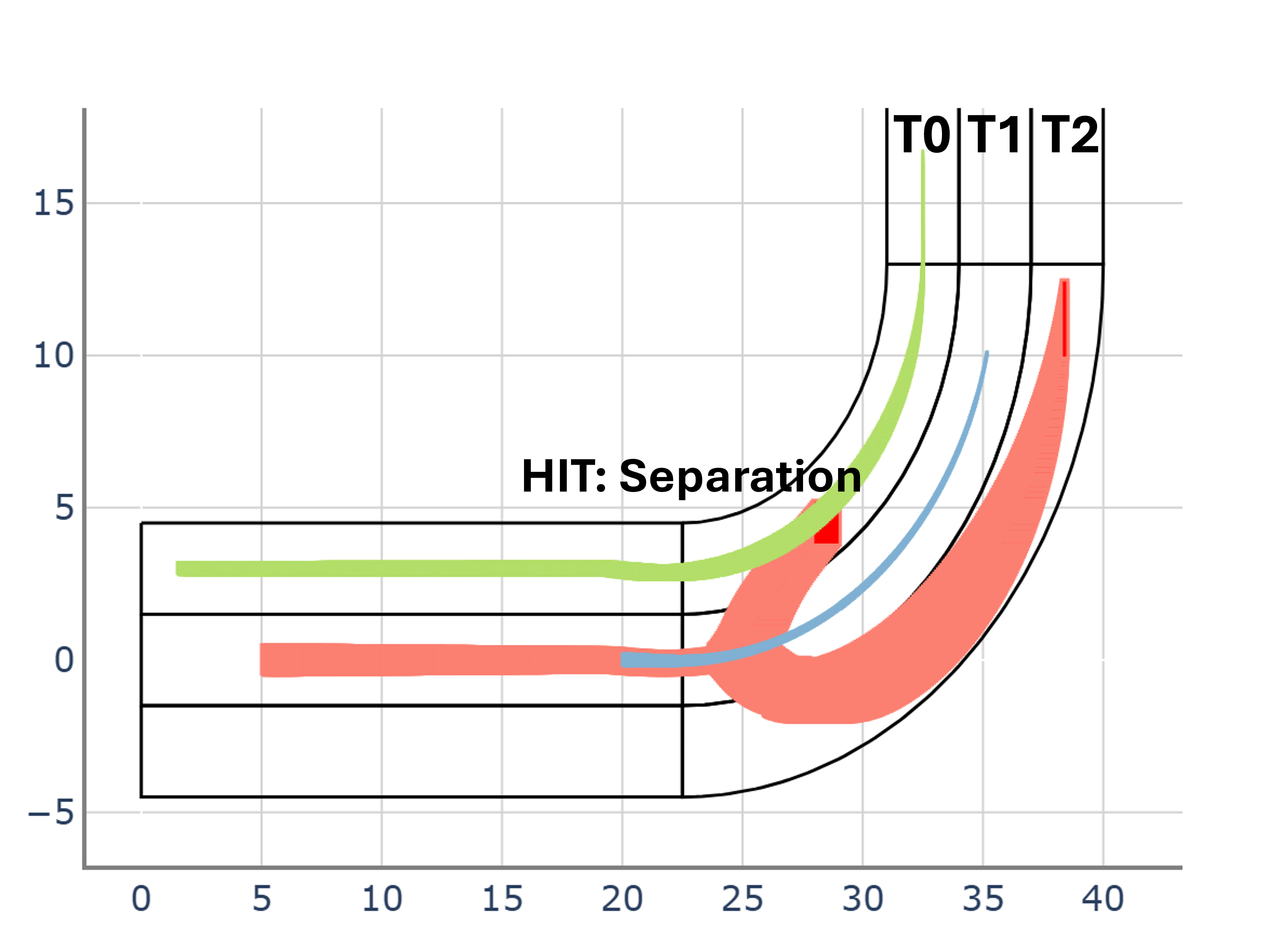}
    \includegraphics[width=0.32\textwidth]{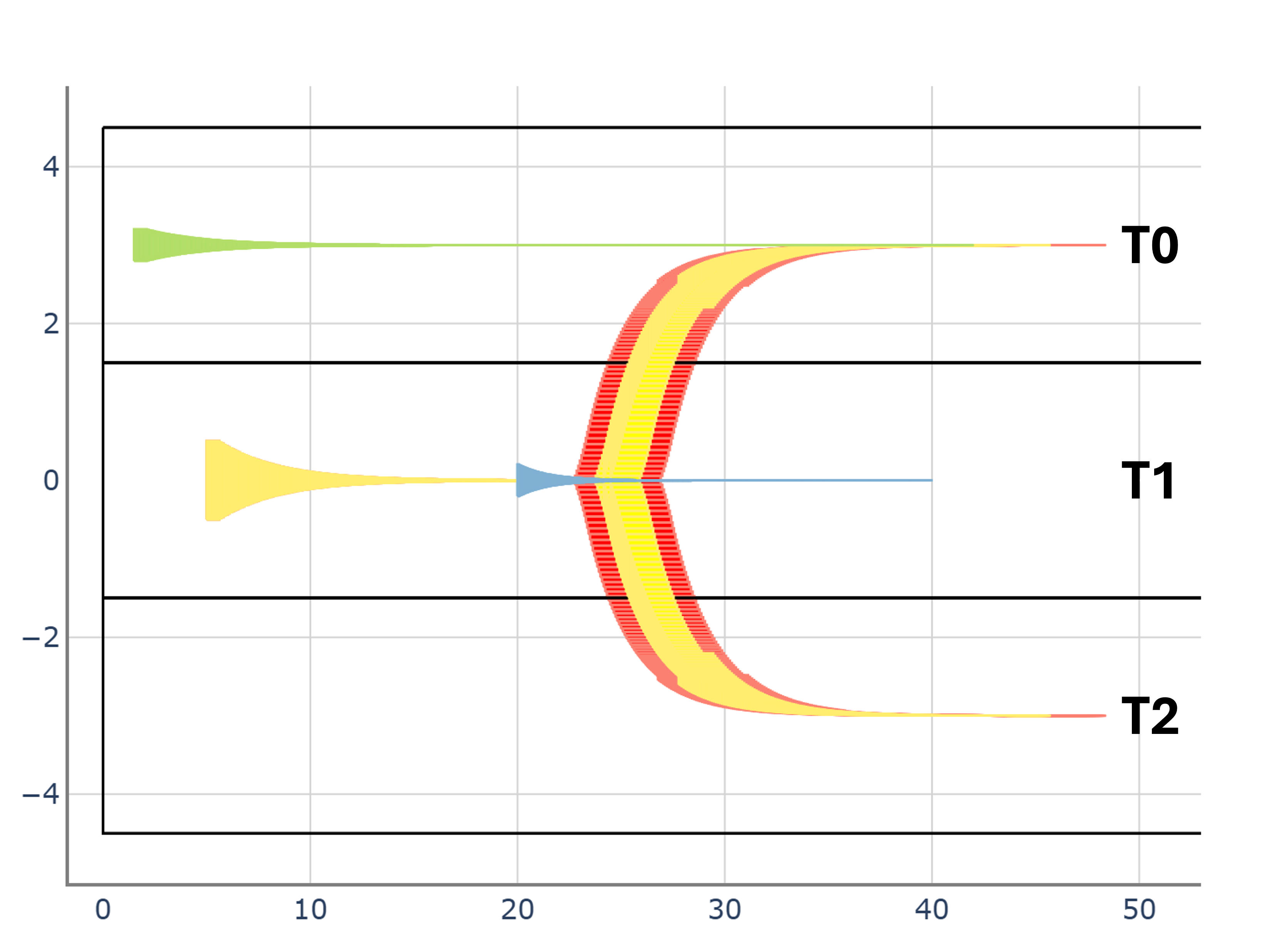}
    \caption{\small \textit{Left:} running the three car scenario on map with parallel straight lanes.  \textit{Center:} same scenario with a curved map. \textit{Right:} same scenario with a noisy sensor.} 
    \label{fig:exp:map_safe}
\end{figure}

\paragraph{Adding noisy sensors.} \ourtool supports scenarios with different sensor functions. For example,  the user can  create a noisy sensor function that mimics a realistic sensor with bounded noise. Such sensor functions are easily added to the scenario using the $\texttt{set\_sensor}$ function.

Fig.~\ref{fig:exp:map_safe} ({\em right}) shows exactly the same three-car scenario with a noisy sensor, which adds $\pm0.5 m$ noise to the perceived position of all other vehicles. 
Since the sensed values of other agents only impacts the checking of the guards (and hence the transitions) of the agents, \ourtool internally bloats the reachable set of positions for the other agents by $\pm0.5$ while checking guards.
Compared with the behavior of the same agent with no sensor noise (shown in yellow in Fig~\ref{fig:exp:map_safe} ({\em right})), the sensor noise enlarges the region over which  the transition can happen, causes enlarged  reachtubes for the red agent.  

\paragraph{Plugging in different reachability engines.}
With a little effort, \ourtool allows users to plug in different reachability tools for the \postcont computation. The user will need to modify the interface of the reachability tool so that given a set of initial states, a mode, and a non negative value $\delta$, the reachability tool can output the set of reachable states over a $\delta$-period represented by a set of timed hyperrectangles.
Currently, \ourtool implements computing $\postcont$ using DryVR~\cite{dryvr}, NeuReach~\cite{neuReach-SunM22} and Mixed Monotone Decomposition~\cite{9304391}. 
A scenario with two car agents in map $\map1$ verified using NeuReach and DryVR is shown in Fig.~\ref{fig:exp:additional} ({\em center}) and ({\em right}).

\begin{figure}
    \centering
    \includegraphics[width=0.32\textwidth]{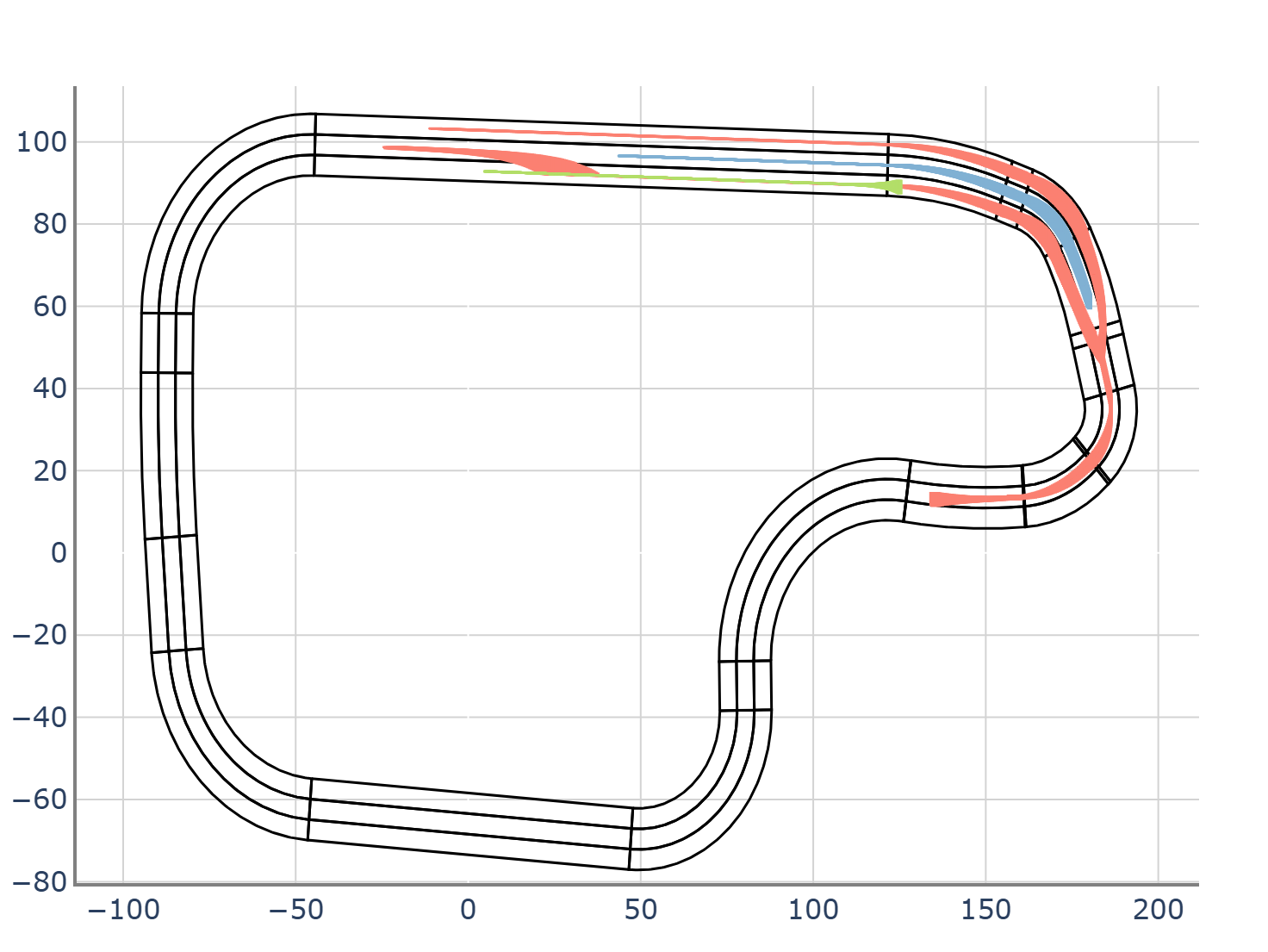}
    \includegraphics[width=0.32\textwidth]{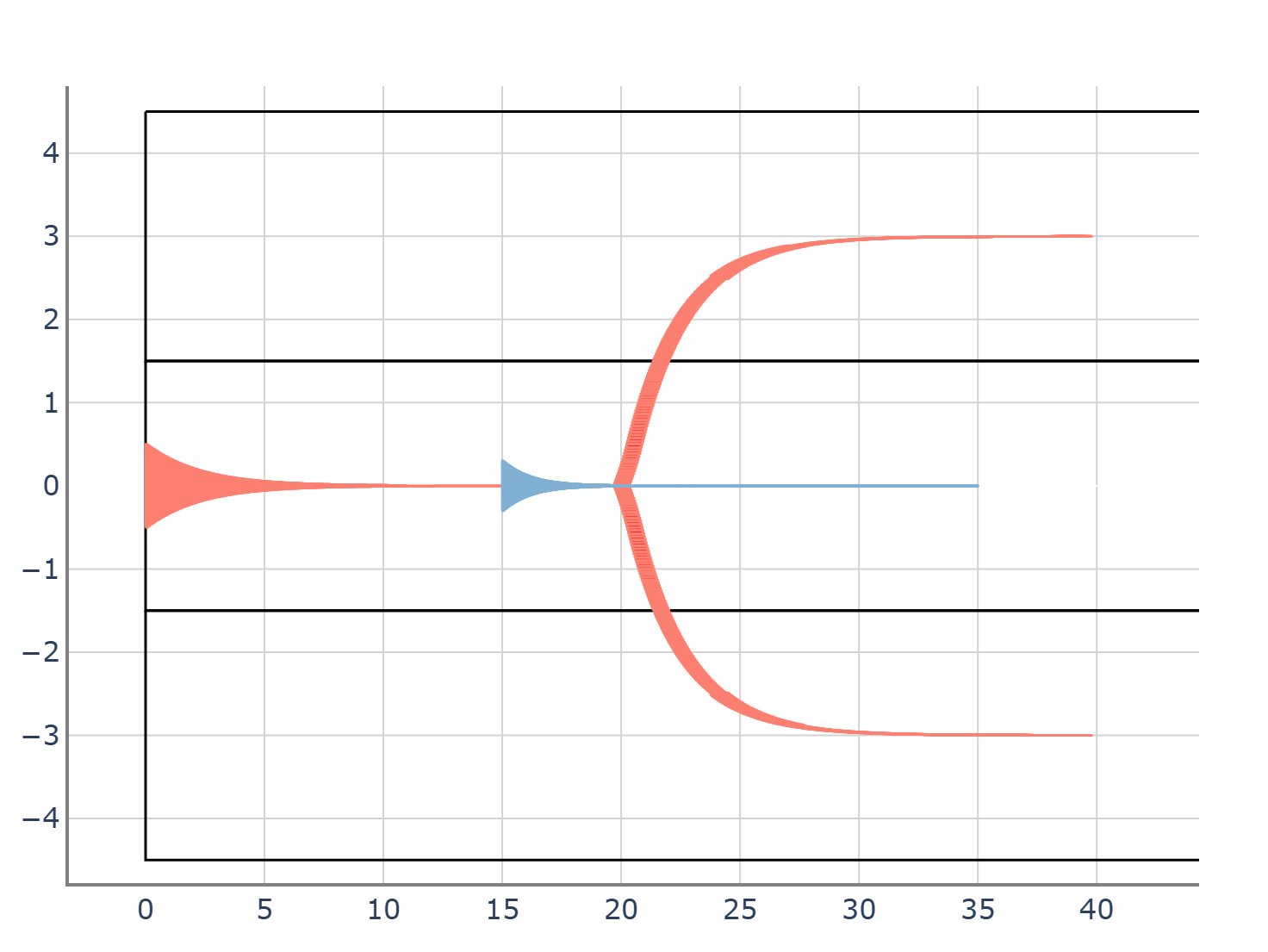}
    \includegraphics[width=0.32\textwidth]{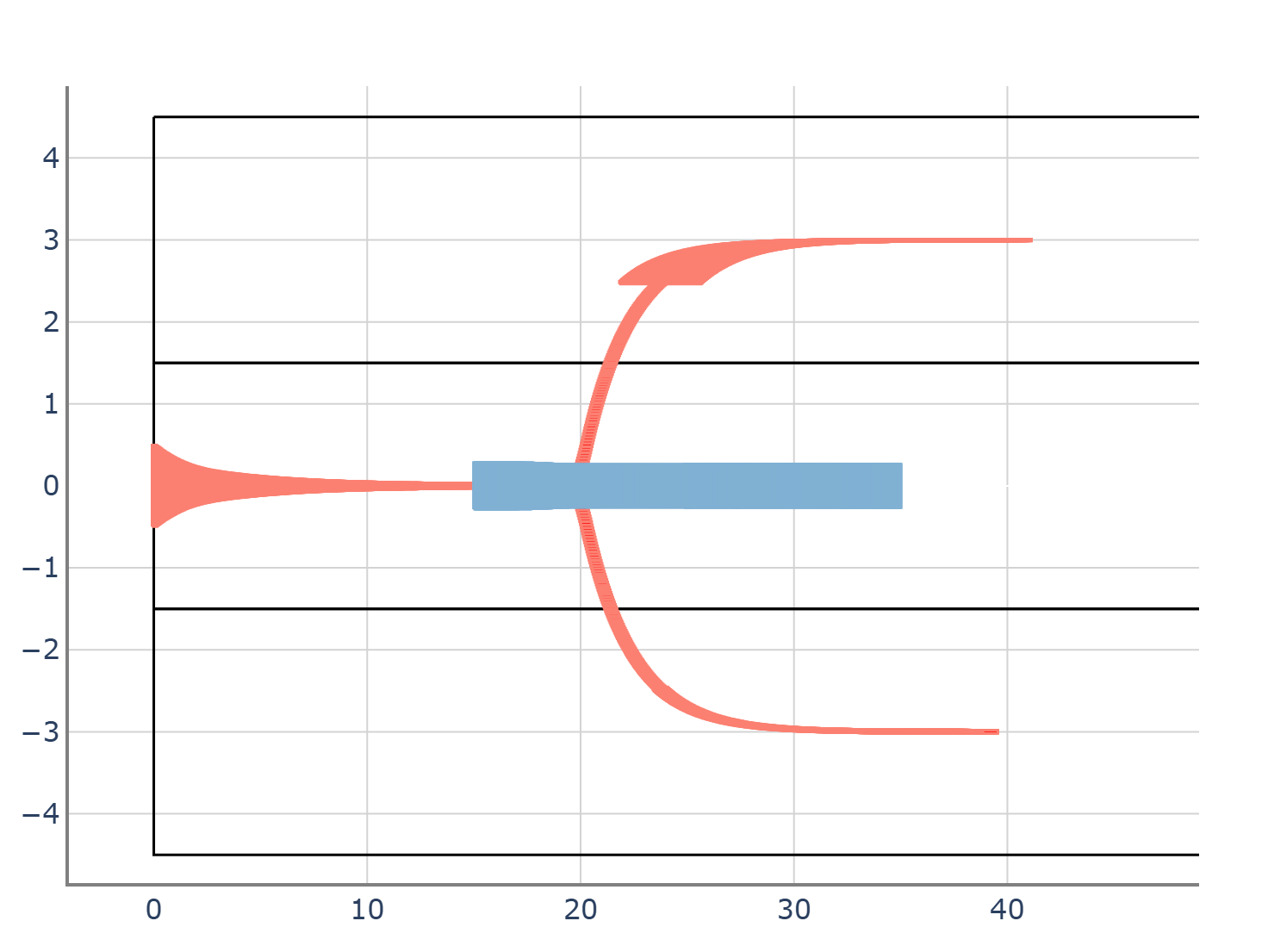}
    \caption{\small \textit{Left:} running the three car scenario on a race track map imported from OpenDrive format.  \textit{Center:} two car scenario verified with DryVR backend. \textit{Right:} same two car scenario verified with NeuReach backend.} 
    \label{fig:exp:additional}
\end{figure}


Table~\ref{tab:exps} reports verification of the examples of this section, together with two standard benchmarks: van-der-pol (Agent V)~\cite{ARCH20:ARCH_COMP20_Category_Report_Continuous_nln} and spacecraft rendezvous (Agent S)~\cite{ARCH20:ARCH_COMP20_Category_Report_Continuous_nln}. Every row was run on one machine with peak memory recorded, and the table extended across the agent $\times$ map $\times$ engine $\times$ sensor grid.

Running time grows with the number of discrete states explored, and the seven-agent scenario on $\map2$, with $37$ states, finishes in $90$\,s. Memory usage is almost flat, $356$--$515$\,MB, drones scenarios use about $100$\,MB above cars.

The choice of reachability engine changes more than the running time. NeuReach\footnote{Runtime for NeuReach includes training time.} is one to two orders of magnitude slower than DryVR, since it trains a network per mode, and a looser tube satisfies guards a tighter one does not: on $\map4$ the three-car scenario explores $7$ states under DryVR and $49$ under NeuReach. Sensing does the same. $\map3$ is unsafe and \ourtool stops at the first failing assertion, which an exact sensor reaches after four states and $\pm0.5$\,m of noise in the first, so that row reports one state rather than a cheaper analysis. Neither \#Tr nor Rt reads without the verdict beside it.

\begin{table}[!htp]
\centering
\ifrevisionmarkup\color{blue}\fi
\setlength{\tabcolsep}{3.5pt}
\caption{\scriptsize \rev{Verification of the examples of Section~\ref{sec:exp}.
Columns: number of agents (\#$\agent$), agent type ($\agent$) --- C (car), D (drone),
V (van der Pol), S (spacecraft), G (gearbox) --- map, reachability engine
($\postcont$), whether the sensor is noisy, number of discrete states explored (\#Tr),
runtime, and peak resident memory. Rows marked $\dagger$ are new in this article and
fill out the agent $\times$ map $\times$ engine $\times$ sensor grid. The row marked
$\ddagger$ was run from the repository branch that produced the conference version,
which carries a controller file absent from the current one. Runtimes are not comparable with those of the
conference version, which used different hardware, though \#Tr reproduces on every row.
Timing is taken around verification itself.}}
\label{tab:exps}
\footnotesize
\begin{minipage}[t]{0.49\linewidth}\centering
\begin{tabular}{r c c l c r r r}
\toprule
\makecell{\#$\agent$} & \makecell{$\agent$} & \makecell{Map} & \makecell{$\postcont$} & \makecell{Noisy\\$\sensor$} & \makecell{\#Tr} & \makecell{Rt\\(s)} & \makecell{Mem\\(MB)} \\
\midrule
2 & C & M1 & DryVR & No & 5 & 3.80 & 360 \\
3 & C & M1 & DryVR & No & 5 & 6.22 & 364 \\
3 & C & M1 & DryVR & Yes & 5 & 7.78 & 359 \\
2 & C & M1 & NeuReach & No & 5 & 204.28 & 515 \\
7 & C & M2 & DryVR & No & 37 & 90.41 & 387 \\
3 & C & M3 & DryVR & No & 4 & 9.15 & 357 \\
3 & C & M4 & DryVR & No & 7 & 41.48 & 361 \\
2 & D & M5 & DryVR & No & 5 & 5.72 & 465 \\
3 & D & M5 & DryVR & No & 7 & 12.16 & 463 \\
2 & D & M5 & NeuReach & No & 5 & 161.03 & 447 \\
\bottomrule
\end{tabular}
\end{minipage}%
\hfill
\begin{minipage}[t]{0.49\linewidth}\centering
\begin{tabular}{r c c l c r r r}
\toprule
\makecell{\#$\agent$} & \makecell{$\agent$} & \makecell{Map} & \makecell{$\postcont$} & \makecell{Noisy\\$\sensor$} & \makecell{\#Tr} & \makecell{Rt\\(s)} & \makecell{Mem\\(MB)} \\
\midrule
2 & D & M6 & DryVR & No & 8 & 16.79 & 468 \\
1 & G$^{\ddagger}$ & N/A & DryVR & N/A & 3 & 5.75 & 386 \\
1 & S & N/A & DryVR & N/A & 3 & 3.88 & 356 \\
1 & V & N/A & DryVR & N/A & 1 & 0.04 & 378 \\
\midrule
3 & C$^{\dagger}$ & M3 & DryVR & Yes & 1 & 3.46 & 364 \\
3 & C$^{\dagger}$ & M3 & NeuReach & No & 1 & 79.62 & 436 \\
3 & C$^{\dagger}$ & M4 & DryVR & Yes & 7 & 47.59 & 365 \\
3 & C$^{\dagger}$ & M4 & NeuReach & No & 49 & 2003.20 & 439 \\
2 & D$^{\dagger}$ & M6 & NeuReach & No & 7 & 412.87 & 476 \\
\bottomrule
\end{tabular}
\end{minipage}
\end{table}

\subsection{Comparison with Other Tools}
In the second part of our evaluation, we quantitatively benchmark \ourtool against two state-of-the-art hybrid systems verification tools, CORA~\cite{Althoff2015a} and JuliaReach~\cite{JuliaReach19}. 
\rev{This comparison is conducted on six standard benchmarks from the ARCH competition~\cite{ARCH23:ARCH_COMP23_Category_Report}, spanning the system classes the formalism is meant to cover. The suite comprises a high-dimensional linear continuous system (ISSC01), a linear hybrid system (GBX01), two nonlinear continuous systems (CVDP23 and LALO20), and two nonlinear hybrid systems (LOVO21 and SPRE22). Dimensions are reported in Table~\ref{tab:comparison}.}

Our assessment focuses on two key aspects: the conservatism of the computed reachable sets and the overall runtime.
\rev{To quantify conservatism we use two size ratios. For a bounded set $\mathcal{S} \subseteq \reals^n$ write $\|\mathcal{S}\| = \sum_{i=1}^{n}\left(\sup_{x \in \mathcal{S}} x_i - \inf_{x \in \mathcal{S}} x_i\right)$ for the $L_1$ norm of the diagonal of its bounding box. The Final-to-Initial ratio (F/I) is $\|\reach{}(T)\| / \|\Xha^0\|$, measuring precision at the end of the horizon, and the Average-to-Initial ratio (A/I) is $\frac{1}{k}\sum_{j=1}^{k}\|\reach{}(t_j)\| / \|\Xha^0\|$ over the $k$ reach sets of the run, measuring overall tightness. We use the bounding-box diagonal rather than volume because the three tools represent reach sets differently---\ourtool by hyperrectangles, CORA by zonotopes, JuliaReach by zonotopes and Taylor models---and the box diagonal is computable from each representation without further over-approximation, so the columns are comparable across tools.} The results of this comparison, detailing the conservatism metrics and runtimes for each tool, are presented in Table~\ref{tab:comparison}.

\begin{table}[!htp]
\centering
\setlength{\tabcolsep}{2.5pt} 
\caption{\scriptsize Comparison of \ourtool with CORA and JuliaReach on standard ARCH benchmarks (e.g., ISSC01, \rev{CVDP23}), which are detailed in Sec.~\ref{sec:exp}. 
The table reports benchmark properties (Dim: dimension, NLN: nonlinear, Hybrid: hybrid system) and performance metrics: Final-to-Initial volume ratio (F/I), Average-to-Initial volume ratio (A/I), and runtime in seconds (Rt). 
The key takeaway is that \ourtool is highly competitive, outperforming CORA and JuliaReach in speed and precision on several nonlinear benchmarks (e.g., \rev{CVDP23}, LALO20).
\rev{The runs in this table are from the conference version and were produced on an AMD~Ryzen~7~5800X with 32\,GB running Windows~10, so their runtimes are not comparable with those reported elsewhere in this section.
\ourtool's A/I entries for the three hybrid benchmarks (GBX01, LOVO21, SPRE22) were recomputed after we found that the original script averaged over the final node of the execution tree rather than over the whole tree; the recomputation reproduces the published F/I values, and the corrected A/I is lower in each case. The single-mode rows are unaffected. Absolute values for SPRE22 vary within the sampling noise of the DryVR engine.}}
\label{tab:comparison} 
\begin{tabular}{l ccc ccc ccc ccc} 
\toprule
\multirow{2}{*}{Benchmark} & \multirow{2}{*}{Dim} & \multirow{2}{*}{NLN} & \multirow{2}{*}{Hybrid} &
\multicolumn{3}{c}{\ourtool} & 
\multicolumn{3}{c}{CORA} & 
\multicolumn{3}{c}{JuliaReach} \\
\cmidrule(lr){5-7} \cmidrule(lr){8-10} \cmidrule(lr){11-13}
 & & & & F/I & A/I & Rt & F/I & A/I & Rt & F/I & A/I & Rt \\
\midrule
ISSC01 & \rev{273} & n & n& 1.518 & 1.926 & 10.84 & 0.651 & 1.255 & 23.73 & 0.834 & 1.067 & 42.05\\
GBX01 & 4 & n & y & 1915.321 & \rev{883.5} & 10.59 & 17.596 & 8.812 & 2.34 & 2.180 & 1.391 & 9.66\\
\rev{CVDP23} & 4 & y & n & 1.075& 0.943& 0.63 & 0.498 & 0.620 & 296.82 & 1.006& 0.941& 8.44 \\
LALO20 & 7 & y & n & 0.028 & 0.383& 0.858 & 0.100& 0.670& 154.02 & 0.051 & 0.570 & 13.42\\ 
LOVO21 & 2 & y & y & 15.114& \rev{10.79}& 5.135 & 19.382 & 8.915 & 17.49 & 1.143 & 1.899 & 3.45 \\
SPRE22 & 4 & y & y & 1.357 & \rev{0.39} & 36.14 & 1.431 & 0.303 & 20.98  &1.104& 0.190 & 117.16 \\
\bottomrule
\end{tabular}
\end{table}

The comparison results in Table~\ref{tab:comparison} first demonstrate the broad applicability of \ourtool, which \rev{solved all six benchmarks}. This suite includes high-dimensional linear systems (ISSC01), linear hybrid systems (GBX01), and a variety of nonlinear continuous (\rev{CVDP23}, LALO20) and nonlinear hybrid (LOVO21, SPRE22) problems.

\rev{\ourtool's advantage is in runtime, and it is largest on the nonlinear benchmarks. On \rev{CVDP23} it finishes in $0.63$\,s against CORA's $296.82$\,s, and on LALO20 in $0.858$\,s against $154.02$\,s; in the high-dimensional linear case it completes ISSC01 more than twice as fast as CORA and four times faster than JuliaReach.}

\rev{Precision is the other side of the trade. \ourtool computes the tightest reachable set of the three tools on LALO20, but on the remaining five benchmarks JuliaReach is tighter, and on ISSC01 and GBX01 CORA is tighter as well. The gap is widest on GBX01, where \ourtool's F/I ratio is $1915.3$ against CORA's $17.6$ and JuliaReach's $2.18$. The cause is representational: \ourtool propagates hyperrectangles, which carry no dependency between state variables, whereas CORA and JuliaReach use zonotopes and Taylor models that do. Each discrete transition in a hybrid benchmark re-approximates the guard intersection by an axis-aligned box, so the error compounds across transitions; GBX01, whose impact resets are frequent, is the worst case. The same mechanism produces the degrading separation bound reported in Section~\ref{sec:exp:multiagent}, and it is what motivates the star-set representation discussed in Section~\ref{sec:conclusion}.}

\subsection{\rev{Sampled-data control}}
\label{sec:exp:sampled}
\begingroup
\ifrevisionmarkup\color{blue}\fi
\rev{Section~\ref{sec:sampleddata} describes two ways of writing a digital controller that samples and holds within a single mode. We measure both on a second-order plant,
\[ \dot{x} = v, \qquad \dot{v} = -K(x_s - r) - Dv, \]
with $K=4$, $D=1$, $r=1$, under a controller running at $100$\,Hz, where $x_s$ is the position as sampled and held over each period. The initial set is $x \in [0, 0.05]$, $v = 0$, the horizon is $1$\,s, and $\delta = 0.002$. We run the plant twice: once with the control applied directly, which makes it a continuous function of the sampled state, and once with the actuator quantized to steps of $0.05$, which makes it a step function of that state. Table~\ref{tab:sampled} reports all cases, with the F/I and A/I ratios taken over $\langle x, v\rangle$.}

\begin{table}[!htp]
\centering
\ifrevisionmarkup\color{blue}\fi
\caption{\scriptsize \rev{A $100$\,Hz controller written with the sampling inside the flow and as an explicit transition per sample, under a continuous and a quantized control law.}}
\label{tab:sampled}
\footnotesize
\begin{tabular}{l l r r r r r}
\toprule
Formulation & Control & Modes & \#Tr & Rt (s) & F/I & A/I \\
\midrule
sampling in $\agentflow_i$ & continuous & 1 & 0   & 0.65  & 1.21   & 1.52 \\
transition per sample      & continuous & 2 & 124 & 10.33 & 105.02 & 38.55 \\
\midrule
sampling in $\agentflow_i$ & quantized  & 1 & 0   & 0.62  & 3.46   & 3.09 \\
transition per sample      & quantized  & 2 & 124 & 10.00 & 127.55 & 46.32 \\
\bottomrule
\end{tabular}
\end{table}

\rev{Quantizing the actuator costs the flow formulation a factor of $2.9$ in F/I, from $1.21$ to $3.46$. This is the penalty for hiding a discontinuity inside $\agentflow_i$: the discrepancy bound must cover the difference between trajectories that fall on opposite sides of a quantization threshold, and it is charged at every sample whether or not the trajectories actually separate. With the continuous control law there is no such threshold and the penalty does not arise.}

\rev{Making each sample a transition costs a factor of $87$ in F/I under the same continuous control law, and is $16$ times slower. Each of the $124$ transitions re-approximates the reachable set, and over a horizon of a hundred periods that accumulation is an order of magnitude larger than the discontinuity penalty it avoids. The ordering is the same for the quantized law. A controller of this kind should therefore be written inside $\agentflow_i$, and the clock formulation reserved for cases where the sampling instants must be reasoned about rather than merely simulated.}
\endgroup

\subsection{\rev{Multi-agent Comparison with JuliaReach}}
\label{sec:exp:multiagent}
\begingroup
\ifrevisionmarkup\color{blue}\fi

We  compare \ourtool with JuliaReach~\cite{JuliaReach19} on two scalable scenarios in which the discrete mode space, rather than the continuous dimension, is what grows. Such a comparison is constrained by what both tools express: \ourtool admits a black-box simulator for $\agentflow_i$, whereas set-propagation tools require symbolic dynamics, so both scenarios keep each agent's flow in closed form.

\paragraph{Benchmark A: braking platoon.}
$N$ vehicles travel in a line, each with longitudinal state $\langle p_i, v_i\rangle$ and three modes: \texttt{Cruise} ($\dot v = 0$), \texttt{Brake} ($\dot v = -2$), and \texttt{Stopped}, entered with the reset $v := 0$. A vehicle brakes when {\em some\/} vehicle ahead lies within $18$\,m, so the guard is an existential quantifier over the other agents. The lead vehicle brakes at time zero and the deceleration propagates down the line as a wave (Fig.~\ref{fig:exp:platoon:wave}). The assertion requires that no vehicle come within $1$\,m of the one ahead. The mode space is $3^N$; no map is involved.

\paragraph{Benchmark B: lane changes on a map.}
Three straight lanes ($\map1$), with $N$ car agents using the bicycle dynamics and Stanley controller of Section~\ref{sec:exp}. A car closing to within $3$ to $5$\,m of a slower car ahead changes lane, its track mode passing through $\mathsf{M_{10}}$ or $\mathsf{M_{12}}$. This benchmark exercises the map abstraction: agent mode and track mode together select the reference path, and admissible changes are those $\maplablab$ permits. Speeds are set so a lane change commands $\arctan(0.45\cdot 3/v) = 0.73$\,rad against the controller's $0.61$\,rad steering limit, so the limit is active during every maneuver.

\paragraph{Method.}
\begin{wrapfigure}[12]{r}{0.38\linewidth}
    \centering
    \ifrevisionmarkup\color{blue}\fi
    \vspace{-1.1\intextsep}
    \includegraphics[width=\linewidth]{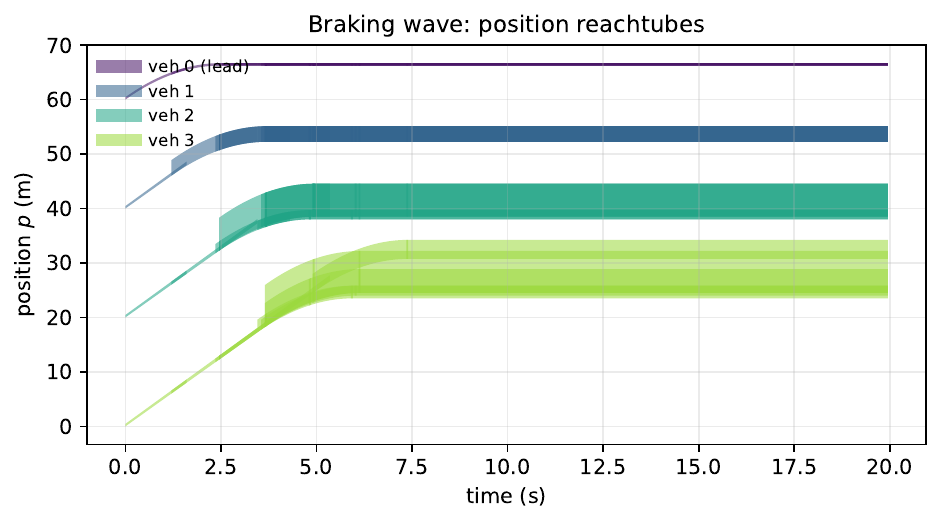}
    \caption{\rev{\small Position reachtubes for the four-vehicle platoon: the vehicles come to rest in sequence, and the tube widens down the line as the uncertainty in each braking time accumulates.}}
    \label{fig:exp:platoon:wave}
\end{wrapfigure}
Both tools are built natively, JuliaReach's compilation time is excluded, and safety checking is inside the timed region for both. JuliaReach is handed the flattened automaton, generated programmatically rather than written by hand, so the comparison measures the size of the object that must exist rather than the effort of typing it. Both tools verify the same assertion and agree on every row of Table~\ref{tab:multiagent}: all scenarios are safe. For Benchmark A we report the tightest separation the analysis certifies between consecutive vehicles, whose true value $11.5$ to $12.9$\,m is known analytically; the ratios of Table~\ref{tab:comparison} are not meaningful here, as the reachable set is a union over many mode branches.

\begin{table}[!htp]
\centering
\ifrevisionmarkup\color{blue}\fi
\setlength{\tabcolsep}{3pt}
\caption{\scriptsize Comparison of \ourtool and JuliaReach on multi-agent scenarios. Columns are: number of agents ($N$), continuous dimension (Dim), and for each tool the discrete states it materializes, runtime (Rt), peak memory (Mem), and the tightest certified separation (Sep, metres; true value $11.5$--$12.9$, larger is tighter). For \ourtool, ``States'' is the number of mode tuples the analysis visits; for JuliaReach it is the number of locations the flattened automaton must contain, with the transition count alongside. Both tools report every scenario safe.}
\label{tab:multiagent}
\begin{tabular}{l r r  r r r r  r r r r r  r}
\toprule
\multirow{2}{*}{Benchmark} & \multirow{2}{*}{$N$} & \multirow{2}{*}{Dim} &
\multicolumn{4}{c}{\ourtool} & \multicolumn{5}{c}{JuliaReach} &
\multirow{2}{*}{$\frac{\text{Rt}_{\text{JR}}}{\text{Rt}_{\ourtool}}$} \\
\cmidrule(lr){4-7} \cmidrule(lr){8-12}
 & & & States & Rt (s) & Mem & Sep & Loc. & Trans. & Rt (s) & Mem & Sep & \\
\midrule
Platoon & 2 & 4 & 4 of 9   & 1.84  & 353 & 11.13 & 9      & 9       & 77.5  & 1327 & 11.88 & 42$\times$ \\
Platoon & 3 & 6 & 7 of 27  & 6.22  & 356 & 7.70  & 27     & 45      & 242.6 & 1519 & 11.90 & 39$\times$ \\
Platoon & 4 & 8 & 11 of 81 & 21.50 & 359 & 4.26  & 81     & 189     & 446.8 & 1756 & 11.97 & 21$\times$ \\
\midrule
Lane change & 2 & 8 & 5 & 5.62 & 356 & --- & 11{,}025 & 122{,}580 & 241.4 & 2763 & --- & 43$\times$ \\
\bottomrule
\end{tabular}
\end{table}

\paragraph{Observations.}
\ourtool is between $21$ and $43$ times faster on every row, and the gap in the discrete state space is larger still: at $N=4$ on the platoon it visits $11$ of $81$ mode tuples, and on the lane-change scenario $5$, where the flattened automaton contains $11{,}025$ locations and $122{,}580$ transitions. \ourtool never constructs the product; the mode space is generated from each agent's decision logic as the analysis explores it.

JuliaReach's analysis is more precise; it certifies a separation of about $11.9$\,m at every $N$, close to the true value, whereas \ourtool certifies $11.13$, $7.70$ and $4.26$\,m as $N$ grows. Because \ourtool composes per-agent hyperrectangles, the uncertainty in one vehicle's braking time is inherited by the vehicle behind it and accumulates along the platoon. This is the same representational limitation behind the conservative ratios of Table~\ref{tab:comparison}.

Two properties of \ourtool's semantics had to be recovered by hand in the flattened models, neither having a counterpart in an automaton whose guards must be convex. A quantified guard such as ``some vehicle ahead is within $18$\,m'' is not convex, so it splits into one transition per sensed agent unless the modeler observes that the platoon ordering makes the immediate predecessor the only relevant neighbour: on the platoon that is the difference between $189$ and $351$ transitions at $N=4$. Urgency requires an invariant negating the guard, again a disjunction and again convex only under the ordering assumption; omitting it in the lane-change model left the branch in which a car does not change lane, and the analysis correctly reported the resulting collision.

\paragraph{Conclusions.}
On scenarios whose difficulty lies in the discrete structure rather than the continuous dimension, \ourtool is an order of magnitude faster and materializes a small fraction of the mode space, at the cost of looser reachable sets that degrade as agents accumulate. Both scenarios also required the modeler to supply, by hand and per benchmark, the split of quantified guards and the invariants that make transitions urgent---exactly the steps the translation of Section~\ref{sec:scenariotoHA} performs from the agent, map and sensor definitions.

\endgroup

\section{Related Works}
\label{sec:related}

Formal verification of hybrid systems has traditionally relied on specialized tools such as SpaceEx~\cite{Spaceex}, Flow*~\cite{ARCH15:Flow_star_1.2_More_Effective}, dReach~\cite{10.1007/978-3-662-46681-0_15}, C2E2~\cite{duggirala2015c2e2}, and CORA~\cite{Althoff2015a}. These operate on low-level modeling languages far removed from software development workflows: SpaceEx takes a custom XML format with explicitly flattened automata, while Flow* and dReach use text-based constraint specifications for ODEs and guards, so developers translate code into formulas by hand. They also lack native support for composition. Inter-agent interaction is modeled by action synchronization over shared labels, and abstractions like sensors or maps are not first-class, so sensor logic and map geometry are encoded directly into transition guards.

A more recent wave of tools has moved closer to the development process by supporting models written in general-purpose programming languages. Tools such as Hylaa~\cite{hylaaBak} and JuliaReach~\cite{JuliaReach19} support systems defined in Python and Julia, respectively, bridging the language gap for differential equation analysis. Similarly, DryVR~\cite{DryVRFanQM017} enables data-driven verification by combining black-box simulators for continuous dynamics (executable in Python) with a white-box transition graph (often specified in JSON) that defines the discrete mode switching structure.

These bridge the language gap but not the architectural one: they view a system as a monolithic set of differential equations or a single global automaton given by a transition graph, rather than a composition of interacting objects, and so lack agent classes, map objects or sensor interfaces. \ourtool treats these as first-class modules, with agents as Python objects interacting through sensor interfaces. It shares the multi-agent focus of SceneChecker~\cite{sibai-tacas-2020}, and adds decision logic in Python together with topological maps that decouple agent intent from geometry.

\rev{\ourtool's continuous post operator is neither a new set representation nor a new integration scheme. $\postcont$ is a pluggable interface whose implementations wrap existing algorithms---DryVR's simulation-driven discrepancy bounds~\cite{dryvr}, NeuReach~\cite{neuReach-SunM22}, and mixed-monotone decomposition~\cite{9304391}---so on a single continuous system \ourtool behaves like whichever is plugged in, and the conservatism in Table~\ref{tab:comparison} reflects their hyperrectangle representation. What is new is the layer above: the abstractions of Section~\ref{sec:scenarios} and the translation of Section~\ref{sec:scenariotoHA}, which build the mode space, guards and resets of a multi-agent automaton from per-agent decision logic, a map and sensors. Section~\ref{sec:exp:multiagent} measures what that buys and costs.}

\rev{A tool paper on \ourtool appeared at CAV 2023~\cite{10.1007/978-3-031-37706-8_18}. It introduced  agents, maps and sensors, the shape of the automaton $\ha(\scenario)$,  and the execution-tree algorithm built from the two post operators, and Section~\ref{sec:exp} reuses the scenarios it reported on. This article provides the detailed  semantics:  any guard of $\ha(\scenario)$ requires a mode change to be admissible under $\maplablab$, which is partial, and map compatibility constrains an agent's transitions; Theorem~\ref{thm:correspondence} relates the behaviors of a scenario and of its automaton. Sensors are set-valued, which is what makes bounded-noise sensing sound, and Proposition~\ref{prop:interleaving} states when the order in which simultaneously enabled agents move cannot matter. Section~\ref{sec:exp:multiagent} compares \ourtool with JuliaReach on two scalable multi-agent scenarios, which the ARCH benchmarks do not exercise.}

Beyond verification, a rich ecosystem supports creating and testing multi-agent scenarios: Scenic~\cite{scenic} defines scenario distributions probabilistically for guided sampling and falsification, Flow~\cite{Wu2017FlowAA} integrates traffic simulators for reinforcement learning, and AAM-GYM~\cite{brittain2022aamgym}, GRAIC~\cite{GRAICrace} and CARLA~\cite{Dosovitskiy17} provide environments for air mobility, racing and driving. These rely on sampling-based simulation or falsification, and so give no symbolic reachability or safety guarantee. \ourtool offers comparable scenario flexibility while retaining set-based analysis.

\section{Discussion and Future Works}
\label{sec:conclusion}

In this paper, we have presented a complete and rigorous mathematical formalism for the \ourtool library, bridging the gap between high-level, user-friendly Python scenarios and formal hybrid systems verification. We established the precise mathematical signatures for the core abstractions—Agents, Maps, and Sensors—and defined the I/O and Map compatibility criteria that formally determine whether a scenario is well-posed. Furthermore, we provided the translation algorithm that compiles these high-level scenarios into hybrid automata, enabling the application of sound reachability analysis. Our experimental evaluation demonstrated that this approach is highly expressive, capable of modeling complex, heterogeneous multi-agent systems, and computationally competitive with state-of-the-art tools.

Despite these advances, the current framework has limitations that point toward important future directions. First, regarding perception, the current Sensor abstraction ($S_i$) is a simplified set-valued function. While it can model basic uncertainty intervals, it does not fully capture the complex, state-dependent errors characteristic of modern, machine-learning-based perception pipelines, nor does it account for environmental factors such as lighting or weather. Second, regarding analysis precision, the current reachability engine typically relies on hyperrectangles or similar representations. As observed in our linear benchmarks, this can lead to conservative approximations compared to tools optimized for specific dynamics, potentially limiting its utility for systems with tight safety margins. \rev{Third, the formalism does not cover agents whose continuous dynamics are coupled. Each $\agentflow_i$ reads only agent $i$'s own state, so interaction is confined to discrete transitions, and a platoon controller that feeds a neighbour's position into the vector field is outside what \ourtool expresses.}

To address these challenges, our future work will focus on two key enhancements. To improve the realism of the sensor model, we plan to integrate perception contracts ~\cite{pmlr-v283-li25c}. This method allows for the rigorous analysis of vision-based control systems by deriving formal contracts for Deep Neural Network (DNN) observers, enabling the automatic identification of Operating Design Domains (ODDs) where system stability is assured. Simultaneously, to address the conservatism in reachability analysis, we aim to incorporate Generalized Star Sets as a core data structure. Preliminary work on "Starred Verse" indicates that this representation can significantly improve precision—reducing reachset diameters by up to 50\% in high-dimensional scenarios—by better capturing the dependencies between continuous state variables.

\begin{acks}
We thank the NASA sponsored ULI research center AVIATE\footnote{\url{https://aviate.illinois.edu/}} at the University of Illinois at Urbana-Champaign for supporting this work. We are also thankful to Bach Viet Hoang for his help with experiments.
\end{acks}

\bibliographystyle{ACM-Reference-Format}
\bibliography{egbib,sayan1}

\end{document}